\documentclass{siamart251216}

\usepackage{amsmath, amsfonts}
\usepackage{mathtools} 
\usepackage{algorithm,algpseudocode}
\usepackage{subcaption}
\usepackage{booktabs}
\usepackage{xspace}
\usepackage{graphicx}
\graphicspath{{Figures/}}

\crefname{appendix}{appendix}{appendices}
\Crefname{appendix}{Appendix}{Appendices}

\newcommand{\psib}{{\boldsymbol{\psi}}}

\newcommand{\omegab}{{\boldsymbol{\omega}}}

\newcommand{\boldf}{{\boldsymbol{f}}}
\newcommand{\ii}{\mathrm{i}} 
\newcommand{\dd}{\mathrm{d}} 
\newcommand{\C}{\mathbb{C}} 
\newcommand{\R}{\mathbb{R}} 
\newcommand{\oCal}{\mathcal{O}} 
\newcommand{\quadrature}{{\mathcal{Q}}}
\newcommand{\oscI}{I} 
\newcommand{\hermiteH}{\mathrm{H}}
\newcommand{\nTimes}{N_\mathrm{T}} 
\newcommand{\nCtrl}{{N_{\rm C}}}
\newcommand{\nHam}{{N_{\rm H}}}
\newcommand{\nFreq}{{N_{\rm F}}}
\newcommand{\schrodinger}{Schr\"odinger\xspace} 
\DeclareMathOperator{\diag}{diag}
\renewcommand{\Re}{\operatorname{Re}} 
\renewcommand{\Im}{\operatorname{Im}}
\newcommand{\tCurr}{{t_{n}}}
\newcommand{\tNext}{{t_{n+1}}}
\newcommand{\tExp}{{t_{n}}}
\newcommand{\tImp}{{t_{n+1}}}
\newcommand{\tExpImp}{{t_{E/I}}} 
\newcommand{\Explicit}{\mathrm{E}}
\newcommand{\Implicit}{\mathrm{I}}
\newcommand{\ExplicitOrImplicit}{*} 
\newcommand{\ExplicitAndImplicit}{{\Explicit/\Implicit}} 
\newcommand{\ctrlFunc}{{c}} 

\newcommand{\dv}[3][]{\frac{d^{#1}#2}{d{#3}^{#1}}}

\newcommand{\ket}[1]{|#1\rangle}

\newsiamremark{remark}{Remark}
\newsiamremark{hypothesis}{Hypothesis}
\crefname{hypothesis}{Hypothesis}{Hypotheses}
\newsiamthm{claim}{Claim}

\ifpdf
  \DeclareGraphicsExtensions{.eps,.pdf,.png,.jpg}
\else
  \DeclareGraphicsExtensions{.eps}
\fi

\ifpdf
\hypersetup{
  pdftitle={Filon Methods for Simulating Highly Oscillatory Controlled Quantum Systems},
  pdfauthor={Spencer Lee, Daniel Appel\"{o}}
}
\fi

\title{
    Filon Methods for Simulating Highly Oscillatory Controlled Quantum Systems
    \thanks{
        Submitted to the editors July 15, 2026.
        \funding{
            This material is based upon work supported by the National Science Foundation (NSF) Graduate Research Fellowship Program under Grant No. 2235783 (S. Lee) and No. DMS-2436319 (D. Appel\"{o}). Part of this research was performed while S. Lee was visiting the Institute for Pure and Applied Mathematics, which is supported by the NSF (Grant Nos. DMS-1925919 and DMS-2422832), and while D. Appel\"{o} was in residence at the Simons Laufer Mathematical Sciences Institute (supported by NSF Grant No. DMS-2424139) in Berkeley, California, during the Fall 2025 semester. D. Appel\"{o} is also supported by the U.S. Department of Energy, Office of Science, Advanced Scientific Computing Research (ASCR), under Award Number DE-SC0025424. Support from Virginia Tech is also acknowledged.
        }
    }
}

\author{
Spencer Lee\thanks{Department of Computational Mathematics, Science, and Engineering \& Department of Physics and Astronomy, Michigan State University, East Lansing, MI, 48824, USA \email{leespen1@msu.edu}.}
\and 
Daniel Appel\"{o}\thanks{Department of Mathematics, Virginia Tech, Blacksburg, VA, 24061, USA \email{appelo@vt.edu}.}
}

\headers{Filon Methods for Oscillatory Quantum Systems}{Spencer Lee and Daniel Appel\"{o}}

\begin{document}
\maketitle

\begin{abstract}
    Fast and accurate classical simulation of quantum systems is a central challenge in the design and control of quantum computers, but the highly oscillatory dynamics of these systems severely limit the efficiency of standard numerical methods. To address this, we adapt Filon quadrature for oscillatory integrals into two numerical methods, called Filon and Controlled Filon, for solving linear systems of ODEs with highly oscillatory solutions. We tailor both methods for efficient implementation in controlled quantum systems, and the Controlled Filon method additionally accounts for the oscillatory structure of the control pulses. We show by numerical experiments that these methods significantly reduce the computational cost of accurately simulating systems of superconducting transmon qubits by decreasing the number of timesteps needed to reach a given level of precision, with only a modest increase in the cost per timestep. For a realistic simulation of the dynamics of a CNOT gate, the Controlled Filon method is the most efficient method tested at every target accuracy, outperforming the best Hermite method by up to $6\times$ and the Hermite method of the same order by up to $500\times$.
\end{abstract}

\begin{keywords}
    Filon, Oscillatory, Numerical Analysis, Quantum Computing, ODEs.
\end{keywords}

\begin{MSCcodes}
    65L05, 65L20, 81-08.
\end{MSCcodes}


\section{Introduction} \label{sec:introduction}
Our ability to realize useful quantum computers depends on our ability to precisely control quantum hardware. Quantum algorithms are composed of quantum gates. For superconducting transmon qubits \cite{Krantz_2019} (and other hardware platforms), gates are implemented by shaping electromagnetic control pulses that drive the desired unitary evolution. Designing these control pulses is a central challenge of quantum optimal control. Analytical pulse shapes exist only for the simplest of gates and hardware models, so in practice the pulse shapes are computationally optimized to minimize an infidelity objective function that measures the accuracy of the implemented gate. The optimizer may compute the objective function (and its gradient) thousands of times, and each time requires numerically simulating the time-dependent \schrodinger equation that governs the quantum system's dynamics, which is the main computational bottleneck of the process. Developing numerical methods that efficiently simulate these dynamics, and thereby evaluate the objective function and its gradient, is a central goal of quantum optimal control \cite{GRAPE1, GRAPE2, CRAB, GRAFS, GOAT, Goerz2022, petersson2022optimal, QUANDARY, petersson2024timeparallel, HOHO_2025}.

The chief obstacle to efficient simulation is the highly oscillatory nature of the dynamics. Transmon transition frequencies are in the GHz range \cite{Krantz_2019}, while multi-qubit gate durations range from tens to hundreds of nanoseconds \cite{ThreeModeTunableCoupler, CrossResonanceGates}, so the system completes hundreds to thousands of oscillations over the course of a single gate. Standard time integrators must take short enough timesteps (on the scale of the shortest period of the solution) to resolve each oscillation \cite{GeometricNumericalIntegration_2006}, which results in a large computational cost. The standard approach to reduce the computational cost is to apply a \emph{rotating frame transformation} (RFT) and make the \emph{rotating wave approximation} (RWA) \cite{Krantz_2019, Rabi1954}. This removes much of the oscillation at the price of increased model approximation error. For coupled transmon qudits (a generalization of the qubit to a system with more than two levels), there is still oscillation even in the rotating frame with the RWA, albeit at a much lower frequency than in the laboratory (lab) frame.

In this work, we adapt Filon quadrature for computing highly oscillatory integrals \cite{IserlesOscillatoryIntegrals2017, filon1930quadrature, IserlesFilonQuadrature_2005} into a timestepping method for linear systems of ODEs, especially those that arise in quantum control. In this approach, the ``natural'' frequencies of the system are deduced from the Hamiltonian and supplied as an ansatz to the method. The resulting scheme becomes \emph{more} accurate as the frequencies \emph{increase}, so that a fixed accuracy can be reached using far fewer timesteps than other methods require. We demonstrate this by comparison with the Hermite (HOHO) method of \cite{HOHO_2025}.

The rest of the paper is structured as follows. In \Cref{sec:problem-description}, we provide a mathematical description of controlled quantum systems, along with a motivating example. In \Cref{sec:numerical-methods}, we review the theory of Filon quadrature for scalar oscillatory integrals (\Cref{sec:scalar-filon-quadrature}), adapt the quadrature into two Filon methods for linear systems of ODEs with highly oscillatory solutions and derive error bounds (\Crefrange{sec:matrix-vector-product-filon}{sec:controlled-filon}), then provide further analysis of the properties and cost of the methods (\Crefrange{sec:prior-work}{sec:quantum-optimal-control}). In \Cref{sec:numerical-experiments}, we perform numerical simulations of a Rabi oscillator and a CNOT gate implementation on a realistic model of two transmon qudits coupled to a resonator bus, comparing the performance of the Filon methods with that of Hermite methods in both the lab and RWA frames. We offer conclusions in \Cref{sec:conclusions}, and provide explicit constructions of the Filon methods in \Cref{app:filon-method-S-formulas,app:controlled-filon-method-S-formulas}.

\section{Problem Description} \label{sec:problem-description}
The dynamics of closed\footnote{The methods described in this paper are also valid for open quantum systems; we restrict ourselves to closed quantum systems only to keep the notation and physics description simple.} quantum systems are governed by \schrodinger's equation:\footnote{We always choose our units so that $\hbar = 1$.}
\begin{equation} \label{eq:schrodinger-equation}
    \frac{d}{dt}\psib(t) = -\ii H(t)\psib(t),\quad 0 \leq t \leq T,\quad \psib(0) = \psib_0 \in \C^N.
\end{equation}
Here, $\psib(t) \in \C^N$ is the \emph{state vector}\footnote{In quantum computing, states are often represented using bra-ket notation, in which case the state $\psib \in \C^N$ corresponds to the state ket $\ket{\psib}$ in the computational basis.}, and the Hermitian matrix $H(t) \in \C^{N \times N}$ is the \emph{Hamiltonian}, which governs the dynamics of the system. Because $H(t)$ is Hermitian, the eigenvalues of $-\ii H(t)$ are purely imaginary, which gives rise to oscillatory behavior of $\psib(t)$.
In the case of controlled quantum systems, the Hamiltonian is decomposed into \emph{drift} and \emph{control} Hamiltonians:
\begin{equation*}
    H(t) = H_d + H_c(t).
\end{equation*}
The drift Hamiltonian governs the ``natural'' dynamics of the system, while the control Hamiltonian governs how the system responds to applied controls, such as microwave pulses. In quantum optimal control, the control Hamiltonian's time dependence typically enters through scalar functions multiplying constant operators \cite{Krantz_2019}:
\begin{equation} \label{eq:scalar-controlled-hamiltonian}
    H(t) = H_d + \ctrlFunc_1(t)H_{c,1} + \cdots + \ctrlFunc_\nCtrl(t)H_{c,\nCtrl}.
\end{equation}
We refer to the scalar functions as \emph{control pulses}.

As an example, consider the Rabi oscillator \cite[\S5]{scully1997quantumoptics} described in the lab and RWA\footnote{``RWA frame'' refers to the system in the rotating frame with the RWA applied.} frames by the Hamiltonians
\begin{equation} \label{eq:rabi-lab-hamiltonian}
    H_\textrm{Lab}(t) = \frac{\omega_0}{2}\sigma_z + E\cos(\omega t)\sigma_x, \quad
    H_\textrm{RWA}(t) = \frac{\Delta}{2}\sigma_z + \frac{E}{2}\sigma_x,
\end{equation}
where $\sigma_z = \bigl[\begin{smallmatrix} 1 & 0 \\ 0 & -1 \end{smallmatrix}\bigr]$ and $\sigma_x = \bigl[\begin{smallmatrix} 0 & 1 \\ 1 & 0 \end{smallmatrix}\bigr]$ are the Pauli-Z and Pauli-X operators, $\omega_0$ is the transition frequency, $E$ is the drive strength, $\omega$ is the drive frequency, and $\Delta \coloneq \omega_0 - \omega$ is the detuning (the RFT is performed using frequency $\omega$). When $E = 0$, the Hamiltonian is time-independent in the lab and RWA frames, and the solution is
\begin{equation} \label{eq:rabi-lab-solution-no-control}
    \psib_{\textrm{Lab}}(t) =
\begin{bmatrix}
e^{-\ii \omega_0 t / 2} \psi_a \\
e^{\ii \omega_0 t / 2}  \psi_b
\end{bmatrix},
\quad \psib_{\textrm{RWA}}(t) =
\begin{bmatrix}
e^{-\ii \Delta t / 2} \psi_a \\
e^{\ii \Delta t / 2}  \psi_b
\end{bmatrix},
\quad \psib(0) \coloneq 
\begin{bmatrix}
    \psi_a \\ \psi_b
\end{bmatrix}.
\end{equation}
If we approximate the solution \eqref{eq:rabi-lab-solution-no-control} numerically, the main challenge is taking enough timesteps to resolve the oscillations of frequency $\pm\omega_0/2$. When $E \neq 0$, no analytic solution is available; if $E \ll \omega_0$, the components of $\psib(t)$ oscillate with frequencies near $\pm\omega_0/2$, so an accurate numerical solution still requires many timesteps.

To demonstrate the utility of Filon-based time integrators, \Cref{fig:rabi-filon-intro} compares the convergence of the Hermite method with that of the Filon and Controlled Filon methods (developed in \Cref{sec:numerical-methods}) when numerically solving \schrodinger's equation for the Rabi oscillator with Hamiltonian \eqref{eq:rabi-lab-hamiltonian}. The numerical experiment is described in more detail in \Cref{sec:rabi-oscillator-experiment}. The Filon and Controlled Filon methods reach the same accuracy as the Hermite methods while using a much larger timestep.

\begin{figure}
    \centering
    \includegraphics[width=\linewidth]{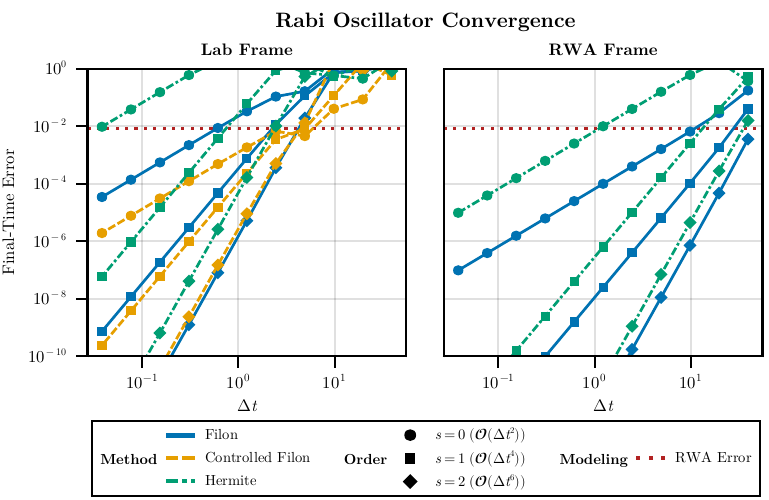}
    \caption{Convergence of the Hermite \cite{HOHO_2025}, Filon (\Cref{sec:filon-for-ODEs}), and Controlled Filon (\Cref{sec:controlled-filon}) methods applied to the Rabi oscillator \eqref{eq:rabi-lab-hamiltonian}. In the RWA frame, the Filon and Controlled Filon methods are identical. The final-time errors are computed using the $\ell^2$ norm.}
    \label{fig:rabi-filon-intro}
\end{figure}

\section{Numerical Methods} \label{sec:numerical-methods}

To approximate the solution to \schrodinger's equation \eqref{eq:schrodinger-equation}, we divide the interval $[0,T]$ into an equidistant grid $t_n = n\Delta t$ for $n=0,\dots,\nTimes$ and $\Delta t \coloneq T / \nTimes$ and let $\psib_n$ denote the numerical approximation to the solution $\psib(t_n)$. To compute $\psib_{n+1}$ from initial data, we put \eqref{eq:schrodinger-equation} in Picard form
\begin{equation} \label{eq:picard-schrodinger}
    \psib(t_{n+1}) = \psib(t_n) + \int_{t_n}^{t_{n+1}}\dv{}{t}\psib(t) \dd t,
\end{equation}
then approximate the integral by quadrature, resulting in a timestepping rule
\begin{equation} \label{eq:S-timestep-rule}
    S_\Implicit \psib_{n+1} = S_\Explicit \psib_n.
\end{equation}
We use the labels $\Implicit$ and $\Explicit$ for the matrices on the left- and right-hand sides of \eqref{eq:S-timestep-rule} because the right-hand side can be computed explicitly, while $\psib_{n+1}$ is computed implicitly by solving a linear system of equations.

In \cite{HOHO_2025}, the integral in \eqref{eq:picard-schrodinger} is approximated using Hermite quadrature (the HOHO method). Our method will be similar, but with the ansatz that the $j$-th component of $\psib(t)$ oscillates at frequency $\omega_j$. That is, if we define
\begin{equation*}
    \omegab \coloneq [\omega_1,\dots,\omega_N]^T, \quad
    R(t) \coloneq e^{-\ii \diag(\omegab) t},
\end{equation*}
then under our highly oscillatory ansatz the solution can be written as
\begin{equation} \label{eq:oscillatory-ansatz}
    \psib(t) \equiv \begin{bmatrix}
    f_1(t)e^{\ii\omega_1 t} \\ \vdots \\ f_N(t)e^{\ii\omega_N t}
    \end{bmatrix}
    \implies \psib(t) = R^\dagger(t)\boldf(t),\quad \boldf(t) \coloneq R(t)\psib(t).
\end{equation}
We call each $e^{\ii \omega_j t}$ and $f_j(t)$ a \emph{carrier wave}\footnote{In general, the carrier wave may take the form $e^{\ii \omega g(t)}$ for some function $g:\R \rightarrow\R$. However, we restrict ourselves to the case of $g(t) = t$, which is sufficient for our application to quantum systems.} and an \emph{envelope}, respectively. We further assume that for $|\omega_j| > 0$, $f_j(t)$ varies more slowly in time than $e^{\ii\omega_j t}$.  

If the components of $\psib(t)$ are highly oscillatory, then the components of the integrand in \eqref{eq:picard-schrodinger} will also be highly oscillatory. Therefore, if we can efficiently evaluate highly oscillatory integrals, then we can efficiently compute solutions to \schrodinger's equation. \Cref{sec:scalar-filon-quadrature} provides an efficient method for scalar integrals, and \Cref{sec:filon-for-ODEs} extends it to the integral in \eqref{eq:picard-schrodinger}.

\subsection{Filon Method for Approximating Integrals} \label{sec:scalar-filon-quadrature}

Consider a highly oscillatory integral of the form
\begin{equation} \label{eq:oscillatory-integral}
\oscI_\omega[f] \coloneq \int_{t_n}^{t_{n+1}} e^{\ii \omega t}f(t) \dd t. 
\end{equation}
We seek an efficient quadrature rule $\quadrature_\omega[f] \approx I_\omega[f]$. Several numerical methods for efficiently computing integrals of this form have been studied in \cite{IserlesOscillatoryIntegrals2017}. Here we use the Filon method from \cite{filon1930quadrature, IserlesFilonQuadrature_2005, IserlesOscillatoryIntegrals2017} that approximates $\oscI_\omega[f]$ by replacing $f(t)$ in the integral \eqref{eq:oscillatory-integral} by a Hermite interpolating polynomial $p_\hermiteH(t)$ whose values \emph{and derivative values} match those of $f(t)$ at multiple nodes in the interval $[t_n,t_{n+1}]$. The resulting quadrature rule is $I_\omega[f] \approx \quadrature_\omega[f] \coloneq I_\omega[p_\hermiteH]$.

We restrict ourselves to the case where the nodes are the endpoints of the interval, and $p_\hermiteH^s(t)$ is the minimal-degree polynomial that matches the values and derivatives of $f(t)$ up to order $s$ at each node.\footnote{A general Hermite interpolating polynomial may match function and derivative values at an arbitrary number of nodes, and with a different number of derivative values matching at each node.} Then $p_\hermiteH^s(t)$ is the unique polynomial of (maximum) degree $2s+1$ satisfying
\begin{equation*}
    \dv[j]{}{t}p_{\hermiteH}^s(t_n) = \dv[j]{}{t}f(t_n),\quad \dv[j]{}{t}p_{\hermiteH}^s(t_{n+1}) = \dv[j]{}{t}f(t_{n+1}),\quad j=0,\dots,s.
\end{equation*}
This Hermite interpolating polynomial $p_\hermiteH^s(t)$ can be constructed using the cardinal polynomials of Hermite interpolation, which are defined by\footnote{Here $\delta_{ij}$ denotes the Kronecker delta, which equals 1 if $i=j$ and equals 0 otherwise.}
\begin{equation*}
    \ell^{(k),s}_{\Explicit,j}(\tExp) = \ell^{(k),s}_{\Implicit,j}(\tImp) = \delta_{jk},\quad
    \ell^{(k),s}_{\Explicit,j}(\tImp) = \ell^{(k),s}_{\Implicit,j}(\tExp) = 0,\quad 0 \leq j,k \leq s.
\end{equation*}
The polynomials $\ell_{\Explicit,j}^s$ and $\ell_{\Implicit,j}^s$ can be found analytically, or their coefficients can be computed by solving a $(2s+2)\times(2s+2)$ linear system. Then $p_\hermiteH^s(t)$ is given by
\begin{equation} \label{eq:hermite-polynomial}
    p_\hermiteH^s(t) \equiv \sum_{j=0}^s \ell^s_{\Explicit,j}(t)f^{(j)}(\tExp) + \ell^s_{\Implicit,j}(t)f^{(j)}(\tImp).
\end{equation}
Replacing $f(t)$ by $p_\hermiteH^s(t)$ in \eqref{eq:oscillatory-integral}, we obtain our quadrature rule $\quadrature_\omega^s[f] \coloneq I_\omega[p_\hermiteH^s]$:
\begin{equation*}
    \quadrature^s_\omega[f] = \sum_{j=0}^s b^s_{\Explicit,j}(\omega) f^{(j)}(\tExp) + b^s_{\Implicit,j}(\omega) f^{(j)}(\tImp),\quad
    b^s_{\ExplicitOrImplicit,j}(\omega) \coloneq\int_{\tCurr}^{\tNext}\ell^s_{\ExplicitOrImplicit,j}(t)e^{\ii\omega t}\dd t.
\end{equation*}

Once the coefficients of each $\ell^s_{*,j}(t)$ are known, the weights $b^s_{*,j}(\omega)$ can be computed in terms of the moments $\mu_m(\omega) \coloneq \int_\tCurr^\tNext t^m e^{\ii \omega t}\dd t$. Computing the moments analytically by integration by parts is unstable for $|\omega| \ll 1$, due to catastrophic cancellation; in that case we instead Taylor expand $e^{\ii\omega t}$ to as many terms as needed and integrate the resulting polynomial.

Furthermore, we can rescale the weights $b^s_{*,j}(\omega)$ so that one set of weights can be accurately and stably computed and then reused for every interval $[\tCurr,\tNext]$. Let $b_{*,j}^{[-1,1],s}(\omega) \coloneq \int_{-1}^{1}\ell^s_{*,j}(t)e^{\ii \omega t}\dd t$. Then the quadrature rule becomes
\begin{multline} \label{eq:scalar-filon-quadrature}
    \quadrature_\omega^s[f] = \sum_{j=0}^s \tilde{b}_{\Explicit,j}^s(\omega) f^{(j)}(\tExp) + \tilde{b}_{\Implicit,j}^s(\omega) f^{(j)}(\tImp),\\
    \tilde{b}_{\ExplicitOrImplicit,j}^s(\omega) \coloneq \left(\frac{\Delta t}{2}\right)^{j+1}e^{\ii\omega(\tCurr+\tNext)/2} \; b_{*,j}^{[-1,1],s} (\hat\omega),\quad \hat\omega \coloneq \omega\frac{\Delta t}{2}.
\end{multline}

\begin{theorem}[Accuracy of Filon Integral Quadrature] \label{thm:scalar-filon-accuracy}
Let $\omega \in \R$ and $f \in C^{2s+2}[t_n,t_{n+1}]$. Then there is a constant $C_s$
depending only on $s$ such that the error in the quadrature rule $\quadrature^s_\omega[f]$ satisfies
\begin{equation} \label{eq:scalar-filon-convergence}
    | I_\omega[f] - \quadrature_\omega^s[f] |
    \leq C_s\,\|f^{(2s+2)}\|_{L^\infty}\;\Delta t^{s+1}
    \min\big(\Delta t^{s+2},\,|\omega|^{-s-2}\big),
\end{equation}
where $\|f^{(2s+2)}\|_{L^\infty}$ is the max norm over $[t_n,t_{n+1}]$.
\end{theorem}
\begin{proof}
\emph{Error in $\Delta t$.} The Hermite interpolation error $e_\hermiteH(t) \coloneq f(t) - p^s_\hermiteH(t)$ and its derivatives are bounded \cite{HermiteInterpolationErrorBirkhoff} by
\begin{equation} \label{eq:hermite-error-bound}
     |e^{(k)}_\hermiteH(t)| \leq M_k \Delta t^{2s+2-k}[u(t)(1-u(t))]^{s+1-k}, \quad k=0,\dots,s+1,
\end{equation}
where $u(t) \coloneq (t-t_n)/\Delta t$ and $M_k \coloneq \|f^{(2s+2)}\|_{L^\infty} / (k!(2s+2-2k)!)$. Because $0 \leq u(1-u) \leq 1$, we have $|e_\hermiteH(t)| \leq M_0 \Delta t^{2s+2}$. Furthermore, because $|e^{\ii \omega t}| = 1$, we have 
\begin{equation} \label{eq:scalar-filon-timestep-bound}
| I_\omega[f] - \quadrature_\omega^s[f] |
= \left|\int_{t_{n}}^{t_{n+1}}e^{\ii\omega t}e_\hermiteH(t)\dd t\right| \leq \int_{t_n}^{t_{n+1}} |e_\hermiteH(t)|\dd t \leq  M_0\Delta t^{2s+3}.
\end{equation}

\emph{Error in $\omega$.} For any $g \in  C^{2s+2}[t_n,t_{n+1}] \subset C^{s+2}[t_n,t_{n+1}]$, $I_\omega[g]$ can be integrated by parts $s+2$ times to give
\begin{equation*}
    I_\omega[g] = 
    -\sum_{k=0}^{s+1} \frac{1}{(-\ii \omega)^{k+1}}
        \left[ g^{(k)}(t_{n+1})e^{\ii\omega t_{n+1}} - g^{(k)}(t_{n})e^{\ii\omega t_{n}}\right]
    + \frac{1}{(-\ii\omega)^{s+2}}I_\omega[g^{(s+2)}].
\end{equation*}
Taking $g = e_\hermiteH = f - p^s_\hermiteH$, 
by the definition of $p^s_\hermiteH$ we have $e_\hermiteH^{(k)}(t_n) = e_\hermiteH^{(k)}(t_{n+1}) = 0$ for $k = 0,\dots,s$. Then the first $s+1$ terms in the finite series are zero, so that
\begin{equation*}
    I_\omega[e_\hermiteH] = \frac{-1}{(-\ii\omega)^{s+2}}[e_\hermiteH^{(s+1)}(t_{n+1})e^{\ii\omega t_{n+1}} - e_\hermiteH^{(s+1)}(t_n)e^{\ii \omega t_n}] + \frac{1}{(-\ii\omega)^{s+2}}I_\omega[e_\hermiteH^{(s+2)}].
\end{equation*}
Taking the absolute value, we get
\begin{equation} \label{eq:filon-frequency-error-bound}
    |I_\omega[e_\hermiteH]| 
    \leq \frac{1}{|\omega|^{s+2}}\left(|e_\hermiteH^{(s+1)}(t_{n+1})| + |e_\hermiteH^{(s+1)}(t_n)| + \int_{t_n}^{t_{n+1}}|e_\hermiteH^{(s+2)}(t)|\dd t\right).
\end{equation}
By \eqref{eq:hermite-error-bound} we have $|e_\hermiteH^{(s+1)}(t_{n})|,\; |e_\hermiteH^{(s+1)}(t_{n+1})| \leq M_{s+1}\Delta t^{s+1}$. Because $e_\hermiteH$ has at least $2s+2$ roots (counting multiplicity) on $[t_n, t_{n+1}]$, Rolle's theorem can be used to show that $e_\hermiteH^{(j)}$ has at least one root $\xi_j \in [t_n, t_{n+1}]$ for $j = 0,\dots,2s+1$. Then 
\begin{equation*}
e_\hermiteH^{(j)}(t) = \int_{\xi_j}^t e_\hermiteH^{(j+1)}(\tau) \dd \tau \implies \|e_\hermiteH^{(j)}\|_{L^\infty} \leq \Delta t \| e_\hermiteH^{(j+1)} \|_{L^\infty}, \quad j=0,\dots, 2s+1.
\end{equation*}
Because $p_\hermiteH^s$ has degree $\leq 2s+1$, we have $e_\hermiteH^{(2s+2)} \equiv f^{(2s+2)}$. Then the above equation implies $\|e_\hermiteH^{(s+2)}\|_{L^\infty} \leq  \Delta t^s \|f^{(2s+2)} \|_{L^\infty}$. Substituting the three bounds into \eqref{eq:filon-frequency-error-bound} gives
\begin{equation*}
    |I_\omega[f] - \quadrature_\omega^s[f]| = |I_\omega[e_\hermiteH]| \leq \left(2M_{s+1} + \|f^{(2s+2)}\|_{L^\infty}\right)\Delta t^{s+1}{|\omega|^{-s-2}}.
\end{equation*}

\emph{Combined Error.} Combining \eqref{eq:scalar-filon-timestep-bound} and
\eqref{eq:filon-frequency-error-bound}, we get 
\begin{equation*}
    |I_\omega[f] - \quadrature_\omega^s[f]|
    \leq \|f^{(2s+2)}\|_{L^\infty}\,\Delta t^{s+1}\min\left(\tfrac{1}{(2s+2)!}\Delta t^{s+2},\;
    \big(1+\tfrac{2}{(s+1)!}\big)|\omega|^{-s-2}\right),
\end{equation*}
which gives \eqref{eq:scalar-filon-convergence} with $C_s \coloneq 1 + 2/(s+1)!$, completing the proof.
\end{proof}

Remarkably, the error \emph{decreases} as the frequency $\omega$ \emph{increases}. Consequently, Filon quadrature can accurately approximate $I_\omega[f]$ for $|\omega| \gg 1$ with far fewer evaluations of $f(t)$ (and its derivatives) than standard quadrature rules use. In the limit $\omega \rightarrow 0$, the number of evaluations required is comparable to that of a standard rule. The method is exact when $f(t)$ is a polynomial of degree $2s+1$ or lower, since in that case $p^s_\hermiteH(t) \equiv f(t)$.

\subsection{Filon Method for Matrix-Vector-Product Integrals} \label{sec:matrix-vector-product-filon}
Now we adapt the quadrature rule $\quadrature_\omega[f]$ to the problem of approximating the integral in \eqref{eq:picard-schrodinger}. For notational convenience, we define the skew-Hermitian matrix $A(t) \coloneq -\ii H(t)$. The skew-Hermitian property of $A(t)$ is not used in this section's derivation, so the method applies, in principle, to any linear system of ODEs (though not all such systems have highly oscillatory dynamics).

In terms of $A(t)$, \schrodinger's equation can be rewritten as
\begin{equation} \label{eq:linear-ode}
    \dv{}{t}\psib(t) = A(t)\psib(t),\quad 0 \leq t \leq T,\quad \psib(0)=\psib_0 ,
\end{equation}
and the integral in the Picard form of \eqref{eq:picard-schrodinger} can be rewritten as
\begin{equation}  \label{eq:picard-integral}
    \int_\tCurr^\tNext \dv{}{t}\psib(t)\dd t = \int_\tCurr^\tNext -\ii H(t)\psib(t)\dd t = \int_\tCurr^\tNext A(t)\psib(t)\dd t \eqcolon I[A,\psib].
\end{equation}
To approximate the matrix-vector-product integral $I[A,\psib]$, define the quadrature rule
\begin{equation} \label{eq:matvec-quadrature-rule} 
    \quadrature_\omegab^s[A,\psib] \coloneq 
    \sum_{j=0}^s \dv[j]{}{t}\bigg(\!A(t) B_{\Explicit,j}^s(\omegab)\boldf(t)\!\bigg)\bigg|_{\tExp} 
    \!\!+ \sum_{j=0}^s \dv[j]{}{t}\bigg(\!A(t) B_{\Implicit,j}^s(\omegab)\boldf(t)\!\bigg)\bigg|_{\tImp}\!,
\end{equation}
where $B^s_{\Explicit,j}(\omegab)$ and $B^s_{\Implicit,j}(\omegab)$ are diagonal weight matrices, defined in terms of the rescaled weights $\tilde{b}^s_{\ExplicitOrImplicit,j}(\omega)$ of \eqref{eq:scalar-filon-quadrature} by
\begin{equation} \label{eq:diagonal-weight-matrices}
B^s_{\ExplicitOrImplicit,j}(\omegab) 
\coloneq \diag( \tilde{b}^s_{\ExplicitOrImplicit,j}(\omega_1), \dots, \tilde{b}^s_{\ExplicitOrImplicit,j}(\omega_N)).
\end{equation}

\begin{theorem}[Filon Matrix-Vector-Product Quadrature] \label{thm:filon-picard-quadrature}
Let $\omegab \in \R^N$, $\boldf(t) \coloneq R(t)\psib(t)$ as in \eqref{eq:oscillatory-ansatz}, and $a_{mk}f_k \in C^{2s+2}[t_n,t_{n+1}]$ for $1 \leq m, k \leq N$. Then the error in the quadrature rule satisfies\footnote{Here $\|\cdot\|_\infty$ denotes the \emph{vector max norm}, $\|\boldsymbol{x}\|_\infty \coloneq \max_i |x_i|$, not the function $L^\infty$ norm $\|\cdot\|_{L^\infty}$.}
\begin{multline} \label{eq:matvec-quadrature-error}
    \| I[A,\psib] - \quadrature^s_\omegab[A,\psib] \|_\infty \\
    \leq C_s\,\Delta t^{s+1} 
    \sum_{k=1}^N \max_{1 \leq m \leq N} \big\|(a_{mk}f_k)^{(2s+2)}\big\|_{L^\infty}
    \min\big(\Delta t^{s+2},\,|\omega_k|^{-s-2}\big),
\end{multline}
with $C_s \coloneq 1 + 2/(s+1)!$ as in \Cref{thm:scalar-filon-accuracy}.
\end{theorem}
\begin{proof}
Applying the integral to each component of \eqref{eq:picard-integral}, we get
\begin{equation} \label{eq:picard-integral-expanded}
    I[A,\psib] \equiv \int_\tCurr^\tNext A(t)\psib(t)\dd t = 
    \sum_{k=1}^N \begin{bmatrix}
        \int_\tCurr^\tNext a_{1k}(t)\psi_k(t) \dd t \\
        \vdots \\
        \int_\tCurr^\tNext a_{Nk}(t)\psi_k(t) \dd t \\
    \end{bmatrix}
    .
\end{equation}
By using the highly oscillatory ansatz \eqref{eq:oscillatory-ansatz}, we see that each integral in \eqref{eq:picard-integral-expanded} is an oscillatory integral of the form \eqref{eq:oscillatory-integral}:
\begin{equation*} 
    I[A,\psib] \equiv \int_\tCurr^\tNext A(t)\psib(t)\dd t = 
    \sum_{k=1}^N \begin{bmatrix}
        \int_\tCurr^\tNext e^{\ii\omega_k t}a_{1k}(t)f_k(t) \dd t \\
        \vdots \\
        \int_\tCurr^\tNext e^{\ii\omega_k t}a_{Nk}(t)f_k(t) \dd t \\
    \end{bmatrix} =
    \sum_{k=1}^N \begin{bmatrix}
         I_{\omega_k}[a_{1k}f_k] \\
        \vdots \\
        I_{\omega_k}[a_{Nk}f_k] \\
    \end{bmatrix}.
\end{equation*}
Approximate each $I_{\omega_k}[a_{mk}f_k]$ with the scalar quadrature $\quadrature^s_{\omega_k}[a_{mk}f_k]$. Then the $m$-th component of $I[A,\psib]$ is approximated by
\begin{align*}
    (I[A,\psib])_m &\approx \sum_{k=1}^N\quadrature^s_{\omega_k}[a_{mk}f_k]  \\
    &= \sum_{k=1}^N \sum_{j=0}^s \left(\!
        \tilde{b}^s_{\Explicit,j}(\omega_k) \dv[j]{}{t}\Big(\!a_{mk}(t)f_k(t)\!\Big)\bigg|_\tExp 
        \!\!\!+ \tilde{b}^s_{\Implicit,j}(\omega_k) \dv[j]{}{t}\Big(\!a_{mk}(t)f_k(t)\!\Big)\bigg|_\tImp
    \right) \\
    &= \sum_{j=0}^s \dv[j]{}{t}\bigg(A(t) B^s_{\Explicit,j}(\omegab)\boldf(t)\bigg)_{\!m}\bigg|_{\tExp} 
        + \sum_{j=0}^s \dv[j]{}{t}\bigg(A(t) B^s_{\Implicit,j}(\omegab)\boldf(t)\bigg)_{\!m}\bigg|_{\tImp}.
\end{align*}
Comparing with the components of \eqref{eq:matvec-quadrature-rule}, we can verify that this approximation of $I[A,\psib]$ is exactly the matrix-vector-product quadrature $\quadrature^s_\omegab[A,\psib]$. Therefore, the error in each component of $\quadrature^s_\omegab[A,\psib]$ is the sum of the errors of $N$ scalar Filon quadratures, which are bounded by \Cref{thm:scalar-filon-accuracy}:
\begin{equation*}
    | (I[A,\psib] - \quadrature_\omegab^s[A,\psib])_m |
    \leq \sum_{k=1}^N C_s\,\|(a_{mk}f_k)^{(2s+2)}\|_{L^\infty}\;\Delta t^{s+1}
    \min\big(\Delta t^{s+2},\,|\omega_k|^{-s-2}\big).
\end{equation*}
Taking the maximum over $m$ gives \eqref{eq:matvec-quadrature-error}, completing the proof.
\end{proof}
\begin{corollary}[Simpler Error Bound of the Matrix-Vector-Product Quadrature] \label{cor:matvec-quadrature-simple-error}
Let $\omega_{\min} \coloneq \min_k |\omega_k|$ and $g^j_{\max} \coloneq \max_{m,k} \|(a_{mk}f_k)^{(j)}\|_{L^\infty}$. Then \eqref{eq:matvec-quadrature-error} implies the simpler bound
\begin{equation*}
    \| I[A,\psib] - \quadrature^s_\omegab[A,\psib] \|_\infty
    \leq N C_s  g^{2s+2}_{\max}\, \Delta t^{s+1} \min\big(\Delta t^{s+2},\,\omega_{\min}^{-s-2}\big).
\end{equation*}
\end{corollary}
\begin{proof}
This follows immediately from \eqref{eq:matvec-quadrature-error}.
\end{proof}
\begin{corollary}[Exactness of the Quadrature] \label{cor:exact-quadrature}
Let each $a_{mk}f_k$ be a polynomial of degree at most $2s+1$ (equivalently, $A\boldf$ is a vector-valued polynomial of degree at most $2s+1$). Then the quadrature rule is exact: $\|\quadrature^s_\omegab[A,\psib] - I[A,\psib]\|_\infty = 0$.
\end{corollary}
\begin{proof}
In this case, $(a_{mk}f_k)^{(2s+2)} \equiv 0$ for all $1 \leq m, k \leq N$. Hence $g^{2s+2}_{\max} = 0$ and \Cref{cor:matvec-quadrature-simple-error} gives $\|\quadrature^s_\omegab[A,\psib] - I[A,\psib]\|_\infty = 0$.
\end{proof}

\subsection{Filon Method for Linear ODEs} \label{sec:filon-for-ODEs}
The solution $\psib(t)$ of the linear system of ODEs \eqref{eq:linear-ode} can be approximated using the timestepping rule
\begin{equation} \label{eq:filon-timestep-rule-quadrature-form}
    \psib_{n+1} \approx \psib_n + \quadrature^s_\omegab[A,\psib].
\end{equation}
Each timestep is performed by solving a system of equations:
\begin{equation} \label{eq:filon-timestep-rule}
    \psib_{n+1} - \sum_{j=0}^s \dv[j]{}{t}\bigg(A(t) B^s_{\Implicit,j}(\omegab)\boldf(t)\bigg)\bigg|_{\tImp} = 
    \psib_{n} + \sum_{j=0}^s \dv[j]{}{t}\bigg(A(t) B^s_{\Explicit,j}(\omegab)\boldf(t)\bigg)\bigg|_{\tExp}.
\end{equation}
We refer to the timestepping rule \eqref{eq:filon-timestep-rule-quadrature-form} as \emph{the Filon method} in the remainder of this work.

\begin{corollary}[Filon Method Local Truncation Error] \label{cor:filon-timestepping}
Let $\omegab \in \R^N$, $\boldf(t) \coloneq R(t)\psib(t)$ as in \eqref{eq:oscillatory-ansatz}, and $a_{mk}f_k \in C^{2s+2}[t_n,t_{n+1}]$ for $1 \leq m, k \leq N$. Then the local truncation error of one timestep of the Filon method \eqref{eq:filon-timestep-rule-quadrature-form} is bounded by 
\begin{multline} \label{eq:filon-ode-local-truncation-error}
   \Big\| \psib(\tNext) - \big( \psib(\tCurr) + \quadrature^s_\omegab[A,\psib] \big) \Big\|_\infty
   \leq N C_s g^{2s+2}_{\max} \Delta t^{s+1} \min(\Delta t^{s+2}, \omega_{\min}^{-s-2}).
\end{multline}
\end{corollary}
\begin{proof}
From the Picard form \eqref{eq:picard-schrodinger}, we have $\psib(\tNext) = \psib(\tCurr) + I[A,\psib]$. Therefore the local truncation error of one timestep is exactly the error in the quadrature $\quadrature^s_\omegab[A,\psib]$, which is bounded by \Cref{thm:filon-picard-quadrature} (and \Cref{cor:matvec-quadrature-simple-error}).
\end{proof}

\begin{remark}[Meaning of asymptotic behavior in $\omegab$] 
Although the frequencies $\omegab$ may be chosen freely, \eqref{eq:filon-ode-local-truncation-error} does not imply that the Filon method's error can be made arbitrarily small by taking $\omega_{\min}$ arbitrarily large. Rather, a poor choice of $\omegab$ will result in a large $g^{2s+2}_{\max}$. Roughly speaking, we have $g^{2s+2}_{\max} = \oCal(\|\delta\omegab\|_\infty^{2s+2})$, where $\delta\omegab$ is the ``mismatch'' between the method's $\omegab$ and the true frequencies of the solution.
\end{remark}

\begin{remark}[Usefulness for controlled systems]
When the drift Hamiltonian is diagonal and no controls are applied ($H(t) \equiv H_d$), the solution is $\psib(t) = e^{-\ii \diag(H_d) t}\psib_0$, and the Filon method is exact when we take $\omegab = -\diag(H_d)$, where here $\diag(H_d)$ denotes the vector of diagonal entries of $H_d$. If we apply control pulses but limit their magnitude, then $H(t)$ is a small perturbation of $H_d$. In that case, we expect the analytical solution will be close enough to the no-control solution that the Filon method with $\omegab = -\diag(H_d)$ will be very accurate. 
\end{remark}

\begin{remark}[Solving the system of equations] 
    The system of equations \eqref{eq:filon-timestep-rule} is linear in $\psib_{n+1}$ and $\psib_{n}$ after substituting $\boldf(t)=e^{-\ii\diag(\omegab) t}\psib(t)$ and expanding the derivatives $\psib^{(j)}(t)$ in terms of $\psib(t)$ using the original ODE system $d\psib(t)/dt = A(t)\psib(t)$ (and its derivatives). Then the system of equations can be written in the form $S^s_\Implicit\psib_{n+1} = S^s_\Explicit\psib_n$. The system can be solved directly by forming $S^s_\Implicit$ and $S^s_\Explicit$, or in a ``matrix-free'' fashion by computing only the vectors $S^s_\Explicit\psib_{n}$ and several candidate solutions $S^s_\Implicit\psib_{n+1}$ in an iterative solver such as GMRES \cite{GMRES}. Explicit formulas for the matrices $S^s_\Implicit$ and $S^s_\Explicit$ for $s=0,1,2$ are given in \Cref{app:filon-method-S-formulas}. 
\end{remark}

\subsection{Modification for Controlled Quantum Systems} \label{sec:controlled-filon}
Consider a controlled quantum system, i.e., a system with a Hamiltonian of the form \eqref{eq:scalar-controlled-hamiltonian}.
When using quantum optimal control to design control pulse shapes that implement desired quantum gates or state-to-state transfers, it is useful to consider control pulses built from one or more carrier waves multiplying smooth envelopes,
\begin{equation} \label{eq:multicarrier-control-pulse}
    \ctrlFunc_k(t) = \sum_{l=1}^{N_{\mathrm{F},k}} \tilde{\ctrlFunc}_{k,l}(t)\, e^{\ii\nu_{k,l}t}.
\end{equation}
Here the $k$-th control Hamiltonian is driven by $N_{\mathrm{F},k}$ carriers, where the $l$-th carrier has envelope $\tilde{\ctrlFunc}_{k,l}(t)$ and carrier frequency $\nu_{k,l}$. The carrier frequencies can be chosen to trigger resonance, which enables fast growth of resonant states and can reduce the time needed to realize a quantum gate or state-to-state transfer; a single control Hamiltonian is often driven by several carriers at once to address several transitions simultaneously. Consequently, there are two sources of highly oscillatory behavior: the implicit oscillation in $\psib(t)$ caused by the dynamics of the system, and the explicit oscillation in the Hamiltonian $H(t)$.

As we have done previously, we simplify the notation by taking $A(t) = -\ii H(t)$ and identifying the terms of $A(t)$ with those of \eqref{eq:scalar-controlled-hamiltonian} via $A_{k} = -\ii H_{c,k}$ for $k=1,\dots,\nCtrl$. The drift Hamiltonian is treated as an additional control Hamiltonian $A_{\nHam} = -\ii H_d$ carrying a single carrier with unit envelope and zero carrier frequency ($N_{\mathrm{F},\nHam} = 1$, $\tilde{\ctrlFunc}_{\nHam,1} \equiv 1$, $\nu_{\nHam,1} = 0$), where $\nHam \coloneq \nCtrl + 1$ is the total number of Hamiltonians. Then we have
\begin{equation} \label{eq:scalar-controlled-hamiltonian-A}
    A(t) = c_1(t) A_1 + \cdots + c_\nHam(t) A_\nHam,
\end{equation}
and the linear system of ODEs describing the time evolution of the controlled quantum system is
\begin{equation} \label{eq:scalar-controlled-schrodinger}
    \dv{}{t}\psib(t) = \Big( \sum_{k=1}^{\nHam} \ctrlFunc_k(t) A_k \Big) \psib(t)
    = \Big( \sum_{k=1}^{\nHam} \sum_{l=1}^{N_{\mathrm{F},k}} \tilde{\ctrlFunc}_{k,l}(t)\, e^{\ii\nu_{k,l}t} A_k \Big) \psib(t).
\end{equation}

For linear ODE systems of the form \eqref{eq:scalar-controlled-schrodinger}, we propose the timestepping rule:
\begin{equation} \label{eq:controlled-filon-quadrature-timestep}
    \psib_{n+1} \approx \psib_n + \sum_{k=1}^{\nHam} \sum_{l=1}^{N_{\mathrm{F},k}} \quadrature^s_{\tilde{\omegab}_{k,l}}[\tilde{\ctrlFunc}_{k,l} A_k,\psib],
\end{equation}
where $\tilde{\omegab}_{k,l} \coloneq [\omega_1 +\nu_{k,l},\ldots, \omega_N + \nu_{k,l}]^T$ denotes the ansatz frequencies $\omegab$ shifted by the $l$-th carrier frequency of the $k$-th control pulse. Each timestep is performed by solving a linear system of equations:
\begin{multline} \label{eq:controlled-filon-timestep-system-of-equations}
    \psib_{n+1} - \sum_{k=1}^{\nHam} A_k \sum_{l=1}^{N_{\mathrm{F},k}} \sum_{j=0}^s B^s_{\Implicit,j}(\tilde{\omegab}_{k,l}) \dv[j]{}{t}\Big(\tilde{\ctrlFunc}_{k,l}(t)\boldf(t)\Big)\bigg|_{\tImp} \\
    = \psib_{n} + \sum_{k=1}^{\nHam} A_k \sum_{l=1}^{N_{\mathrm{F},k}} \sum_{j=0}^s B^s_{\Explicit,j}(\tilde{\omegab}_{k,l}) \dv[j]{}{t}\Big(\tilde{\ctrlFunc}_{k,l}(t)\boldf(t)\Big)\bigg|_{\tExp} .
\end{multline}
We refer to this method as the \emph{Controlled Filon method}. As with the Filon method, the Controlled Filon method can be written in the form $S^s_\Implicit \psib_{n+1} = S^s_\Explicit \psib_n$, and this linear system of equations can be solved directly or using an iterative solver such as GMRES. Explicit formulas for the Controlled Filon matrices $S^s_\Implicit$ and $S^s_\Explicit$ ($s=0,1,2$) are given in \Cref{app:controlled-filon-method-S-formulas}.

\begin{corollary}[Controlled Filon Method Local Truncation Error] \label{cor:controlled-filon}
Let $\omegab \in \R^N$, $\boldf(t) \coloneq R(t)\psib(t)$ as in \eqref{eq:oscillatory-ansatz}, and $\tilde{\ctrlFunc}_{k,l}f_j \in C^{2s+2}[t_n,t_{n+1}]$ for $1 \leq k \leq \nHam$, $1 \leq l \leq N_{\mathrm{F},k}$, $1 \leq j \leq N$. Then with $C_s$ as in \Cref{thm:scalar-filon-accuracy}, the local truncation error of one timestep of the Controlled Filon method \eqref{eq:controlled-filon-quadrature-timestep} is bounded by
\begin{multline} \label{eq:controlled-filon-local-truncation-error}
    \bigg\| \psib(\tNext) - \Big( \psib(\tCurr) + \sum_{k=1}^{\nHam} \sum_{l=1}^{N_{\mathrm{F},k}} \quadrature^s_{\tilde{\omegab}_{k,l}}[\tilde{\ctrlFunc}_{k,l}A_k,\psib] \Big) \bigg\|_\infty \\
    \leq C_s\,\Delta t^{s+1} \!\sum_{k=1}^{\nHam} \! \sum_{l=1}^{N_{\mathrm{F},k}} \! \sum_{j=1}^N
         \big\|(\tilde{\ctrlFunc}_{k,l} f_j)^{(2s+2)}\big\|_{L^\infty} 
         \! \min\!\big(\Delta t^{s+2},\,|\omega_j+\nu_{k,l}|^{-s-2}\big)\max_{m} \big|(A_k)_{mj}\big|.\!
\end{multline}
\end{corollary}
\begin{proof}
Moving the summations over the control Hamiltonians and their carriers outside of the integral $I[A,\psib]$, where $A(t)$ has the structure of \eqref{eq:scalar-controlled-schrodinger}, we get
\begin{equation*}
    I[A,\psib]
    = \sum_{k=1}^{\nHam} \sum_{l=1}^{N_{\mathrm{F},k}}
        \int_\tCurr^\tNext\tilde{\ctrlFunc}_{k,l}(t)A_k e^{\ii\diag(\tilde\omegab_{k,l})t}\boldf(t)\dd t
    = \sum_{k=1}^{\nHam} \sum_{l=1}^{N_{\mathrm{F},k}}
        I_{\tilde\omegab_{k,l}}[\tilde{\ctrlFunc}_{k,l} A_k, \psib].
\end{equation*}
If we approximate each $I_{\tilde{\omegab}_{k,l}}[\tilde{\ctrlFunc}_{k,l} A_k,\psib]$ with the Filon quadrature $\quadrature^s_{\tilde{\omegab}_{k,l}}[\tilde{\ctrlFunc}_{k,l} A_k,\psib]$, then the Picard form $\psib(t_{n+1})=\psib(t_n) + I[A,\psib]$ becomes \eqref{eq:controlled-filon-quadrature-timestep}. The linear system of equations \eqref{eq:controlled-filon-timestep-system-of-equations} results directly from the definition of $\quadrature^s_{\tilde{\omegab}_{k,l}}[\tilde{\ctrlFunc}_{k,l} A_k,\psib]$ and using the fact that the time dependence of $A(t)$
is limited to the control pulses $\ctrlFunc_k(t)$, which commute with each matrix $A_k$. The local truncation error of one timestep is exactly the sum of the errors in approximating each $I_{\tilde{\omegab}_{k,l}}[\tilde{\ctrlFunc}_{k,l}A_k,\psib]$ with $\quadrature^s_{\tilde{\omegab}_{k,l}}[\tilde{\ctrlFunc}_{k,l}A_k,\psib]$. Applying \Cref{thm:filon-picard-quadrature} to each term and factoring out the time-independent $(A_k)_{mj}$ from each $L^\infty$ norm  gives \eqref{eq:controlled-filon-local-truncation-error}.
\end{proof}

\begin{remark}[Behavior of the Controlled Filon Method] The error bound \eqref{eq:controlled-filon-local-truncation-error} is governed by the smallest \emph{modified} frequency $\min_{j,k,l}|\omega_j + \nu_{k,l}|$, and the method is exact when each $\tilde{\ctrlFunc}_{k,l}(t) f_j(t)$ is a polynomial of degree $2s+1$ or lower (which follows directly from \Cref{cor:exact-quadrature}).
\end{remark}

\subsection{Comparison to Prior Work} \label{sec:prior-work}
A Filon-type quadrature for highly oscillatory systems of ODEs was developed by Khanamiryan \cite{Khanamiryan2008, Khanamiryan2011} for problems of the form $\dot{\mathbf{y}} = A_\omega(t)\mathbf{y} + \mathbf{f}(t,\mathbf{y})$, where $A_\omega(t)$ has large, purely imaginary eigenvalues. In their formulation, variation of parameters is used to write the solution as $X_\omega(t)\mathbf{y}_0 + \int X_\omega(t-\tau)\mathbf{f}(\tau, \mathbf{y}) \dd \tau$, where the fundamental matrix $X_\omega$ satisfies a matrix linear ODE $\dot{X}_\omega = A_\omega X_\omega$. Like our method, theirs becomes more accurate as the frequency grows.

When applying their method to a controlled \schrodinger equation, $A_\omega(t)$ can be taken to include either the entire Hamiltonian or only the drift term $-\ii H_d$, in which case the control Hamiltonians are incorporated into the forcing term $\mathbf{f}(t,\boldsymbol\psi) = -\ii H_c(t)\psib$. In the former case, the method reduces to a Magnus-expansion integrator \cite{Blanes2009Magnus}, and Filon quadrature is not used. In the latter case, $\int X_\omega(t-\tau)\mathbf{f}(\tau, \mathbf{y}) \dd \tau$ is discretized using Filon quadrature, i.e., $\mathbf{f}$ (which includes the highly oscillatory solution $\psib$) is replaced by a Hermite interpolating polynomial, and the resulting system of equations for each timestep is solved by waveform-relaxation iteration.

In contrast, we attribute the oscillation directly to the solution via $\psi_j = f_j e^{\ii \omega_j t}$ \eqref{eq:oscillatory-ansatz} using frequencies inferred from the drift Hamiltonian, discretize $\int -\ii H(t) \psib \dd t$ using Filon quadrature with $\psib$ factored into slowly varying envelope and highly oscillatory carrier-wave parts, and replace only the envelope $\boldf$ by a Hermite interpolating polynomial. We additionally account for explicit oscillation in $H(t)$ via the Controlled Filon method (\Cref{sec:controlled-filon}). Each timestep is performed via a single linear solve. 

\subsection{Stability Analysis for the Second- and Fourth-Order Methods} \label{sec:stability}
Consider the Filon method applied to the Dahlquist equation $\frac{dy}{dt} = \lambda y$,  
\begin{equation*}
    y_{n+1} = y_n + \lambda\int_{t_n}^{t_{n+1}} p^s_\hermiteH(t)e^{\ii\omega t}\dd t,  \end{equation*}
where
\begin{equation*}
        p^s_\hermiteH(t) = \sum_{j=0}^s \ell^s_{\Explicit,j}(t)f^{(j)}(t_n) + \sum_{j=0}^s \ell^s_{\Implicit,j}(t)f^{(j)}(t_{n+1}).
\end{equation*}
We write $f_n \coloneq f(t_n)$ and $f_{n+1} \coloneq f(t_{n+1})$, so that $y_n = f_n e^{\ii\omega t_n}$. With $\kappa \coloneq \lambda - \ii \omega$, we can write
\begin{equation*}
    \frac{dy}{dt} = \lambda y,\quad y = e^{\ii \omega t}f \implies \frac{df}{dt} = \kappa f \implies f^{(j)} = \kappa^j f.
\end{equation*}
Expressing $p^s_\hermiteH(t)$ in terms of $f$ we get 
\begin{equation*}
    p^s_\hermiteH(t) = f_n\sum_{j=0}^s \ell^s_{\Explicit,j}(t)\kappa^j + f_{n+1}\sum_{j=0}^s\ell^s_{\Implicit,j}(t)\kappa^j.
\end{equation*}
To compute the weights from the integrals we parameterize
\begin{equation} \label{eq:t-parameterization}
    t = t_n + \sigma\Delta t, \qquad \sigma \in [0,1],
\end{equation}
and  introduce
\begin{equation*}
    z \coloneq \lambda\Delta t, \qquad \varphi \coloneq \omega\Delta t, \qquad \hat\kappa \coloneq (\lambda - \ii\omega)\Delta t = \kappa\Delta t = z - \ii\varphi.
\end{equation*}
When $\omega = 0$, the method reduces to the Hermite method, which is A-stable \cite{Hagstrom_HermiteInTime}, so we assume $\omega \neq 0$ (hence $\varphi \neq 0$) in what follows. Note that $\lambda\kappa^j\Delta t^{j+1} = z\hat\kappa^j$. With the change of variables \eqref{eq:t-parameterization}, the explicit term is
\begin{equation*}
    \int_{t_n}^{t_{n+1}} e^{\ii\omega t}\ell^s_{\Explicit,j}(t) \dd t = \Delta t^{j+1} e^{\ii\omega t_n} \int_0^1 e^{\ii\varphi\sigma}\hat\ell^s_{\Explicit,j}(\sigma) \dd\sigma, \quad \hat\ell^s_{\ExplicitOrImplicit,j}(\sigma) \coloneq \Delta t^{-j}\ell^s_{\ExplicitOrImplicit,j}(t_n+\sigma\Delta t),
\end{equation*}
where $\hat\ell^s_{\ExplicitOrImplicit,j}(\sigma)$ is the cardinal polynomial on the reference interval $[0,1]$. Then
\begin{equation*}
    \lambda\int_{t_n}^{t_{n+1}} \Big(\sum_{j=0}^s\ell^s_{\Explicit,j}(t)\kappa^j f_n\Big)e^{\ii\omega t}\,\dd t = z\,y_n\sum_{j=0}^s\hat\kappa^j\int_0^1 e^{\ii\varphi\sigma}\hat\ell^s_{\Explicit,j}(\sigma)\,\dd\sigma.
\end{equation*}
For the implicit part we use the same parameterization and note that
\begin{equation*}
    f_{n+1}e^{\ii\omega t_{n+1}} = f_{n+1}e^{\ii\omega t_n}e^{\ii\omega\Delta t} = f_{n+1}e^{\ii\omega t_n}e^{\ii\varphi},
\end{equation*}
so $f_{n+1}e^{\ii\omega t_n} = e^{-\ii\varphi}y_{n+1}$. Similar manipulations give
\begin{equation*}
    \lambda\int_{t_n}^{t_{n+1}}\Big(\sum_{j=0}^s\ell^s_{\Implicit,j}(t)\kappa^j f_{n+1}\Big)e^{\ii\omega t}\,\dd t = z\,y_{n+1}\sum_{j=0}^s\hat\kappa^j\int_0^1 e^{\ii\varphi(\sigma-1)}\hat\ell^s_{\Implicit,j}(\sigma)\,\dd\sigma.
\end{equation*}
Denote the explicit coefficients 
\begin{equation*}
    \gamma_j = \int_0^1 e^{\ii\varphi\sigma}\hat\ell^s_{\Explicit,j}(\sigma)\,\dd\sigma.
\end{equation*}
Then we can exploit $\hat\ell^s_{\Implicit,j}(\sigma) = (-1)^j\hat\ell^s_{\Explicit,j}(1-\sigma)$ to find the implicit coefficients
\begin{equation*}
    \int_0^1 e^{\ii\varphi(\sigma-1)}\hat\ell^s_{\Implicit,j}(\sigma)\,\dd\sigma
    = \int_0^1 e^{\ii \varphi(\sigma-1)}(-1)^j \hat\ell^s_{\Explicit,j}(1-\sigma)\dd\sigma
    = (-1)^j\overline{\gamma_j}.
\end{equation*}
If we write the method as $Q(z)y_{n+1} = P(z)y_n$ we thus have
\begin{equation*}
    P(z) = 1 + z\sum_{j=0}^s \gamma_j\hat\kappa^j, \qquad Q(z) = 1 - z\sum_{j=0}^s (-1)^j\overline{\gamma_j}\,\hat\kappa^j.
\end{equation*}

The step amplification factor is $y_{n+1}/y_n = P(z)/Q(z)$, so for a given $s$ A-stability amounts to showing that $|P(z)| \le |Q(z)|$ for every $\omega$ and all $\lambda$ with $\Re(\lambda) \le 0$. We now establish the property for $s = 0$ and $s = 1$.

When $s=0$, we have $\hat\ell^0_{\Explicit,0}(\sigma) = 1-\sigma$ and $\hat\ell^0_{\Implicit,0}(\sigma) = \sigma$. Let $\alpha \coloneq \ii\varphi$ and $\beta \coloneq e^{\ii\varphi} - 1 - \ii\varphi$. Then a direct computation yields
\begin{align*}
    \int_0^1(1-\sigma)e^{\ii\varphi\sigma}\dd\sigma = \frac{e^\alpha - 1 - \alpha}{\alpha^2}  = -\frac{\beta}{\varphi^2}, \quad
    \int_0^1 \sigma e^{\ii\varphi(\sigma-1)}\dd\sigma = -\frac{\overline \beta}{\varphi^2},
\end{align*}
so that $P(z) = 1 - z\beta/\varphi^2$ and $Q(z) = 1 + z\overline \beta/\varphi^2$.
Since $\varphi^2 > 0$, the inequality $|P(z)| \le |Q(z)|$ is equivalent to $|\varphi^2 P(z)| \le |\varphi^2 Q(z)|$, and
\begin{align*}
    |\varphi^2 P(z)|^2 &= (\varphi^2 - z\beta)(\varphi^2 - \bar z\overline \beta) = \varphi^4 - 2\varphi^2\Re(z\beta) + |z|^2|\beta|^2, \\
    |\varphi^2 Q(z)|^2 &= (\varphi^2 + z\overline \beta)(\varphi^2 + \bar z \beta) = \varphi^4 + 2\varphi^2\Re(z\overline \beta) + |z|^2|\beta|^2.
\end{align*}
Subtracting these,
\begin{equation*}
    |\varphi^2 P(z)|^2 - |\varphi^2 Q(z)|^2 = -2\varphi^2\Re(z\beta + z\overline \beta) = -4\varphi^2\Re(\beta)\Re(z).
\end{equation*}
Since $4\varphi^2 > 0$ and
\begin{equation*}
    \Re(\beta) = \Re(e^{\ii\varphi}-1) = \cos(\varphi) - 1 \leq 0,
\end{equation*}
the difference $|\varphi^2 P|^2-|\varphi^2 Q|^2$ is nonpositive when $\Re(z) = \Re(\Delta t \lambda) \leq 0$, so $|P(z)| \le |Q(z)|$. Equality occurs when $\lambda$ is purely imaginary. Finally, note that $P/Q$ has no pole in the left half-plane, so the amplification factor is analytic there:
\begin{equation*}
    \Re(z_{\mathrm{pole}}) = \frac{\varphi^2}{|\beta|^2}(1 - \cos\varphi) \ge 0.
\end{equation*}

To study the case $s=1$ we note that $Q(z) = \overline{P(-\bar z)}$. Then, if the zeros of $P(z)$ are $z_k$ we can write $P(z) = \gamma_{s} \prod_{k=0}^{s}(z-z_k)$, and with  $Q(z) = \overline{P(-\bar z)} = (-1)^{s+1} \bar{\gamma}_{s} \prod_{k=0}^{s}(z+\bar{z}_k)$ we see that the zeros of $Q(z)$ are those of $P(z)$ mirrored around the imaginary axis. This reduces the stability question, $|P(z)/Q(z)| \leq 1, \Re(z) \leq 0$, to simply checking that the zeros of $P(z)$ lie in the left half-plane (so that the zeros of $Q(z)$ lie in the right half-plane and $P/Q$ has no poles for $\Re(z) \leq 0$). One way to check this is by the Routh--Hurwitz criterion for polynomials with complex coefficients \cite{frank1946zeros}. For the case $s=1$, writing $P(z) = 1 + w_1 z + w_2 z^2$, the Routh--Hurwitz stability conditions reduce to checking that
\begin{equation*}
  \Re(w_1\bar{w}_2) \geq 0,  \ \ \ \ \Im(w_2)^2 \leq \Re(w_1) \Re(w_1 \bar{ w}_2).
\end{equation*}
The coefficients $w_1$ and $w_2$ can be computed explicitly. Recalling $\alpha = \ii\varphi$, we have
\begin{align*}
  w_2 &= \int_0^1 e^{\alpha \sigma}\sigma(1-\sigma)^2\,\dd\sigma  = \frac{(2\alpha-6)e^{\alpha}+\alpha^2+4\alpha+6}{\alpha^4}, \\
  w_1 &= \int_0^1 e^{\alpha \sigma}(1-\sigma)^2(1+2\sigma)\,\dd\sigma-\alpha w_2 =  \frac{(-2\alpha^2+12\alpha-12)e^{\alpha}-2\alpha^3-4\alpha^2+12}{\alpha^4}.
\end{align*}
Using that $\alpha^4=(\ii\varphi)^4 = \varphi^4$ and $e^{\alpha}=\cos\varphi+\ii\sin\varphi$, we can separate the real and imaginary parts as $w_1 = (R_1+\ii I_1)/\varphi^4$ and $w_2 = (R_2+\ii I_2)/\varphi^4$, where
\begin{alignat*}{2}
  R_1 &= 4\varphi^2+12+(2 \varphi^2 - 12) \cos \varphi -12 \varphi \sin \varphi, \quad &
  R_2 &= 6-\varphi^2 - 6 \cos \varphi - 2 \varphi \sin \varphi, \\
  I_1 &= 2 \varphi^3 + 12 \varphi \cos \varphi + (2\varphi^2-12) \sin \varphi, \quad &
  I_2 &= 4 \varphi + 2\varphi \cos \varphi - 6 \sin \varphi.
\end{alignat*}

Having introduced $R_1$, $R_2$, $I_1$, and $I_2$, the stability conditions reduce to 
\begin{equation*}
G_1 = R_1 R_2 + I_1 I_2 \geq 0, \ \ G_2 = R_1 G_1 - \varphi^4 I_2^2 \geq 0.
\end{equation*}

The equation for $G_1$ can be written $G_1 = a + b \cos \varphi + c \sin \varphi,$ with
\begin{equation*}
   a = 4\varphi^4 + 24\varphi^2 + 144,\ \ b = 2 \varphi^4 + 48\varphi^2 - 144,\ \ c = - 144\varphi.
\end{equation*}
Using that $|b \cos \varphi + c \sin \varphi | \le \sqrt{b^2+c^2}$, and $a>0$, we can check that
$a^2-b^2-c^2 \geq 0$ instead of checking $G_1 \geq 0$. A lengthy calculation shows that
\begin{equation*}
a^2-b^2-c^2=12\varphi^2(\varphi^6+1728)-20736\varphi^2 = 12\varphi^{8}  \geq 0.
\end{equation*}
To check the second condition, $G_2 \geq 0$, we used SymPy \cite{SymPy} to simplify $G_2$ into
\begin{equation*}
 G_2 =108  (\varphi^2+4)^2 \left(1 + \cos(\varphi+\theta) \right)^2,
\end{equation*}
with $\theta$ defined by
\begin{equation*}
  \cos \theta = \frac{\varphi^2-4}{\varphi^2+4}, \qquad \sin \theta = \frac{4\varphi}{\varphi^2+4}.
\end{equation*}
In principle, higher $s$ can be checked in the same way, but we did not pursue this.

\subsection{Global Error} \label{sec:global-error}
    For a stable one-step method, the global error is controlled by the accumulated local truncation error; the error at time $t \leq T$ is at most a constant multiple of the sum of the one-step errors of each timestep \cite{GeometricNumericalIntegration_2006}. Hence, the local bounds on the error for the Filon (\Cref{cor:filon-timestepping}) and Controlled Filon (\Cref{cor:controlled-filon}) methods give 
\begin{equation} \label{eq:filon-global-error}
    \| \psib_n - \psib(t_n) \|_\infty
    = \oCal\Big(\Delta t^{s} \min\big(\Delta t^{s+2}, \omega_{\min}^{-s-2}\big) \Big),
\end{equation}
    with $\omega_{\min}$ replaced by the smallest modified frequency $\min_{j,k,l}|\omega_j + \nu_{k,l}|$ for the Controlled Filon method. Hence the Filon and Controlled Filon methods have order $2s+2$, conditional on stability. We have proven A-stability for $s = 0, 1$, and hypothesize that the methods remain A-stable for $s > 1$. Because the Hamiltonian has real eigenvalues, $-\ii H(t)$ has purely imaginary eigenvalues, so A-stability implies that the methods are stable at any $\Delta t$ for time-independent Hamiltonians, and hence converge with global order $2s+2$. When the Hamiltonian \emph{does} depend on time, stability is not guaranteed, but in the numerical experiments of \Cref{sec:numerical-experiments}, we do not observe any instability that limits the convergence of the methods.

\subsection{Computational Cost} \label{sec:computational-cost}
We consider the case of a controlled Hamiltonian \eqref{eq:scalar-controlled-hamiltonian-A}, and assume that each timestep is performed by using GMRES to solve the linear system of equations $S^s_\Implicit \psib_{n+1} = S^s_\Explicit \psib_n$. The cost of each timestep taken by the Filon, Controlled Filon, and Hermite methods is dominated by applications of $S^{s}_{\ExplicitAndImplicit}$. We can compare the three methods' costs for any value of $s$ in terms of the number of linear algebra operations required to apply $S^{s}_{\ExplicitAndImplicit}$. We group the operations into \emph{GEMV} ($\boldsymbol{y} \gets \alpha M \boldsymbol{x} + \beta \boldsymbol{y}$, i.e., applications of constant Hamiltonians $A_k$), \emph{diagonal matrix-vector products} (the same with a diagonal matrix, e.g., $\Omega \coloneq \diag(\omegab)$ or the diagonal matrices $G^{s}_{\ExplicitAndImplicit,k,m}$ defined in \eqref{eq:efficient-filon-S-formula} and \eqref{eq:controlled-filon-G-definition}), and \emph{AXPY} ($\boldsymbol{y} \gets \alpha \boldsymbol{x} + \boldsymbol{y}$, excludes the fused $\beta \boldsymbol{y}$ additions in the other operation types). The latter two operations are $\oCal(N)$, so if the matrices $A_k$ are dense we expect that the $(s+1)\nHam$ GEMVs will dominate the cost. In many quantum systems of interest, including the superconducting transmon system studied in \Cref{sec:multi-qudit-gate-problem}, the matrices $A_k$ are sparse ($1$--$2$ nonzero entries per row), so the cost of the $\oCal(N)$ operations has a larger relative impact.

The operation counts can be derived from the explicit formulas in \Cref{app:filon-method-S-formulas,app:controlled-filon-method-S-formulas}, and are tabulated in \Cref{tab:matvec-counts}. We base our counts on the efficient Filon and Controlled Filon methods, defined in \eqref{eq:efficient-filon-S-formula} and \eqref{eq:explicit-controlled-filon-efficient}, taking advantage of all opportunities to reuse computed results (e.g., reusing $F_m\psib$, and only computing $G^s_{\ExplicitAndImplicit,k,m}$ once per timestep). All three methods require the same number of GEMV operations, so we expect the per-timestep cost to be comparable across the methods.
\begin{table}[htb!]
    \centering
    \setlength{\tabcolsep}{3pt}
    \begin{tabular}{lccc}
        \toprule
        \multicolumn{4}{c}{Operation Counts per Application of $S^s_{\ExplicitAndImplicit}$ ($S^s_{\Implicit}$ Once per GMRES iteration)} \\
        \midrule
        & Hermite & \multicolumn{2}{c}{Filon \& Controlled Filon} \\
        \midrule
        GEMV & $(s+1)\nHam$ & \multicolumn{2}{c}{$(s+1)\nHam$} \\
        Diagonal Mat-Vec & 0 & \multicolumn{2}{c}{$\binom{s+1}{2}+(s+1)\nHam$} \\
        AXPY & $1 + s\nHam + \binom{s+1}{2}\nHam$ & \multicolumn{2}{c}{$1 + \binom{s+1}{2}\nHam$} \\
        \midrule
        \multicolumn{4}{c}{Operation Counts in Assembly of $G^s_{\ExplicitAndImplicit,k,m}$ (Once per Timestep, per Side)} \\
        \midrule
        & Hermite & Filon & Controlled Filon \\
        \midrule
        AXPY & 0 & $\binom{s+2}{2}\nHam$ & $\binom{s+2}{2}\nFreq$ \\
        \bottomrule
    \end{tabular}
    \caption{Operations needed to apply $S^s_{\ExplicitAndImplicit}$ in the Filon and Controlled Filon methods for the case of a controlled Hamiltonian \eqref{eq:scalar-controlled-hamiltonian-A}, using the efficient forms \eqref{eq:efficient-filon-S-formula} and \eqref{eq:explicit-controlled-filon-efficient}. $\nFreq \coloneq \sum_k N_{\mathrm{F},k}$ is the total number of carrier waves in the control pulses.
    }
    \label{tab:matvec-counts}
\end{table}

\subsection{Application to Quantum Optimal Control} \label{sec:quantum-optimal-control}
In the HOHO method \cite{HOHO_2025} for quantum optimal control problems, the Hamiltonian is parameterized ($H(t) = H(t;\boldsymbol{\theta})$, $\boldsymbol{\theta} \in \R^{N_{\mathrm{Param}}}$) and the gradient of the objective function is computed exactly and efficiently (the number of differential equations solved numerically does not scale with $N_{\mathrm{Param}}$) by using a discrete adjoint approach.

In \cite{HOHO_2025}, the discrete adjoint for the Hermite method was derived using the fact that timesteps take the form $S_\Implicit \psib_{n+1} = S_\Explicit \psib_n$ (using this work's notation). Because timesteps in the Filon and Controlled Filon methods take the same form (but with different matrices), the HOHO method's discrete adjoint also applies to them, differing only in how $S_{\ExplicitAndImplicit}^\dagger$ and $\partial S_{\ExplicitAndImplicit} / \partial \theta_k$ are applied in the adjoint evolution and gradient accumulation. As with the Hermite method, the resulting scheme will compute the gradient exactly and efficiently without relying on automatic differentiation (and hence will not require memory allocation at each timestep). Implementing this discrete adjoint is beyond the scope of this work.

\section{Numerical Experiments} \label{sec:numerical-experiments}
\subsection{Rabi Oscillator} \label{sec:rabi-oscillator-experiment}
In \Cref{fig:rabi-filon-intro}, we compare the convergence of the Hermite, Filon, and Controlled Filon methods with $s = 0$, $1$, and $2$ for the Rabi oscillator problem using parameters $\omega_0 = 1$, $\omega = 0.9$, $E = 0.01$ in \eqref{eq:rabi-lab-hamiltonian}. We simulate for duration $T = 200\pi / \omega_0$, which corresponds to 100 periods of the transition frequency $\omega_0$ (50 oscillations of each solution component) in the lab frame. The drive is near-resonant but detuned ($\omega = 0.9\omega_0$), which causes oscillation in the RWA frame, but at a lower frequency than in the lab frame. We simulate in both the lab and RWA frames.

For each method, we compute a coarse numerical solution with a large $\Delta t$, then iteratively refine it by halving $\Delta t$. For each simulation, we compute the error in the numerical solution at the final time against a reference solution computed by a high-order adaptive Runge--Kutta method \cite{Verner2010, Rackauckas2017} with absolute and relative tolerances of $10^{-15}$. We also use $10^{-15}$ as the absolute and relative tolerances when solving the linear system at each timestep using GMRES for the Hermite, Filon, and Controlled Filon methods. We use no preconditioning, and take $\psib_n$ as the initial guess for $\psib_{n+1}$.

\Cref{fig:rabi-filon-intro} reports the results in both the lab and RWA frames. Each method converges at the expected order $\oCal(\Delta t^{2(s+1)})$, confirming the error analysis of \Cref{sec:filon-for-ODEs,sec:controlled-filon}. Moreover, the Filon and Controlled Filon methods reach the same target accuracy as the Hermite method of the same order while using substantially larger timesteps. The advantage is greatest in the lab frame and at low order: to reach an error of $10^{-2}$, the $s=0$ Filon and Controlled Filon methods use timesteps roughly $17\times$ and $160\times$ larger, respectively, than the $s=0$ Hermite method. At $s=1$, the Filon and Controlled Filon methods achieve a given accuracy using timesteps roughly $3\times$ and $4\times$ larger than Hermite, and roughly $2\times$ larger at $s=2$.

We also mark the RWA modeling error with a horizontal line. Even an exact solution of the RWA dynamics, transformed back into the lab frame, differs from the exact lab-frame solution by an $\ell^2$ final-time error of roughly $10^{-2}$. Therefore, any apparent improvement in accuracy in the RWA frame below the $\approx 10^{-2}$ threshold is not an improvement with respect to the true dynamics. In the lab frame, the methods converge to the true solution, which is why the Filon and Controlled Filon methods are valuable: they make solving the exact lab-frame problem affordable, rather than requiring the RWA.

\subsection{Multi-Qudit Gate Problem} \label{sec:multi-qudit-gate-problem}
We now focus on a larger and more practical problem: implementing a CNOT gate on two qudits coupled to a resonator bus, a problem that has been studied in \cite{petersson2022optimal, HOHO_2025}. The lab-frame drift Hamiltonian for $N_Q$ transmon qudits in the dispersive limit is
\begin{equation*}
    H_d = \sum_{q=1}^{N_Q}\left(
        \omega_q a_q^\dagger a_q - \frac{\xi_q}{2}a_q^\dagger a_q^\dagger a_q a_q - \sum_{p > q} \xi_{pq}a_p^\dagger a_p a_q^\dagger a_q 
    \right), \quad
    H_c(t) = \sum_{q=1}^{N_Q} c_q(t)\left(a_q + a_q^\dagger\right),
\end{equation*}
where $\omega_q$ and $\xi_q$ are, respectively, the ground state transition frequency and self-Kerr coefficient (anharmonicity) of qudit $q$, and $\xi_{pq}$ is the cross-Kerr coefficient between qudits $p$ and $q$. Each operator $a_q$ is the lowering operator of the subsystem corresponding to qudit $q$:
\begin{align*}
    a_q \coloneq  I_{n_{N_Q}} \otimes \cdots \otimes I_{n_{q+1}} \otimes L_{n_q} \otimes I_{n_{q-1}} \otimes \cdots \otimes I_{n_1}, \quad N \coloneq \prod_{q=1}^{N_Q} n_q, \\
    L_n \coloneq \begin{bmatrix}
    0 & \sqrt{1} &        &            \\
      & \ddots   & \ddots &            \\
      &          & \ddots & \sqrt{n-1} \\
      &          &        & 0          \\
    \end{bmatrix}
    \in \R^{n \times n},\quad
    I_n \coloneq \begin{bmatrix}
        1 & & \\
        & \ddots & \\
        & & 1 
    \end{bmatrix} \in \R^{n \times n},
\end{align*}
where $n_q$ is the number of levels modeled for qudit $q$.

Driving transitions between the qudit levels requires delivering energy near the qudit transition frequencies, so the lab-frame control functions take the ansatz
\begin{equation*}
    c_q(t) \coloneq 2 \Re(d_q(t) e^{\ii \omega_q t}) \equiv d_q(t)e^{\ii \omega_q t} + \overline{d}_q(t)e^{-\ii \omega_q t}, \quad d_q(t) \coloneq c_{I,q}(t) + \ii c_{Q,q}(t).
\end{equation*}

When modeling two qudits coupled to a resonator bus, the resonator is treated as another qudit, with a higher transition frequency and smaller self-Kerr coefficient. We label the resonator with $R$ to distinguish it from qudits $1$ and $2$. Then the lab-frame drift Hamiltonian is
\begin{multline*}
    H_d = \omega_1 a_1^\dagger a_1 + \omega_2 a_2^\dagger a_2  + \omega_R a_R^\dagger a_R
        -\frac{\xi_1}{2} a_1^\dagger a_1^\dagger a_1 a_1 -\frac{\xi_2}{2}  a_2^\dagger a_2^\dagger a_2 a_2 -\frac{\xi_R}{2} a_R^\dagger a_R^\dagger a_R a_R \\
        - \xi_{21} a_2^\dagger a_2 a_1^\dagger a_1 - \xi_{R1}  a_R^\dagger a_R a_1^\dagger a_1 - \xi_{R2} a_R^\dagger a_R a_2^\dagger a_2.
\end{multline*}

The values used for the physical parameters in this experiment are provided in \Cref{tab:cnot3-physical-parameters}, and were originally obtained from \cite{petersson2022optimal}. The qudits are coupled weakly to one another, but more strongly to the resonator, which facilitates a two-qudit gate.

\begin{table}[htb!] 
\centering
\begin{tabular}{lccccc}
  \toprule
  \multicolumn{6}{c}{Physical Parameters of CNOT System} \\
  \midrule[\heavyrulewidth]
  Subsystem $q$ & $\omega_q$ & $\xi_{q1}$ & $\xi_{q2}$ & $\xi_{qR}$ & $n_{q}$ \\
  \midrule
  Qudit $1$ & 4.11 & 2.20(-1) & -- & -- & 4  \\
  Qudit $2$ & 4.82 & 1.00(-6) & 2.25(-1) & -- & 4  \\
  Resonator & 7.84 & 2.49(-3) & 2.52(-3) & 2.83(-5) & 10 \\
  \bottomrule
\end{tabular}
\caption{Physical parameters (rounded to three significant digits) for two qudits coupled to a resonator bus. The values for each $\omega$ and $\xi$ are given as $\omega/2\pi$ and $\xi/2\pi$ in GHz, and $\xi_{qq} \coloneq \xi_q $ is the self-Kerr coefficient of subsystem $q$.}
\label{tab:cnot3-physical-parameters}
\end{table}

The diagonal operators $\omega_q a_q^\dagger a_q$ in $H_d$ cause rapid oscillation in each component of $\psib$, with frequency increasing with the excitation level of qudit $q$. Applying the rotating frame transformation $U(t) \coloneq \exp(\ii t \sum_q \omega_q a_q^\dagger a_q)$ gives the Hamiltonian
\begin{gather*}
\tilde H_d = -\sum_{q=1}^{N_Q}\left(
    \frac{\xi_q}{2}a_q^\dagger a_q^\dagger a_q a_q + \sum_{p>q}\xi_{pq}a_p^\dagger a_p a_q^\dagger a_q
\right), \\
    \tilde{H}_c(t) = \sum_{q=1}^{N_Q}c_{I,q}(t)(a_q+a_q^\dagger) + \ii c_{Q,q}(t)(a_q-a_q^\dagger) + d_q(t)e^{2\ii \omega_q t}a^\dagger_q + \overline{d}_q(t)e^{-2\ii\omega_q t}a_q.
\end{gather*}
In the RWA, we drop the counter-rotating terms that oscillate at $\pm2\omega_q$.

We now numerically simulate the system of two qudits coupled to a resonator bus under the influence of control pulses that implement a CNOT gate. We use the pulses found in \cite{HOHO_2025}, which were optimized to a gate fidelity of 99.99\% (verified in the lab frame). The gate duration is 550 nanoseconds, and the control functions $d_q(t)$ are degree-14 B-spline curves multiplied by carrier waves with frequencies $0$, $-2\pi\xi_1$, and $-2\pi \xi_2$ for qudits 1 and 2, and frequencies $0$, $-2\pi \xi_{R1}$, and $-2\pi \xi_{R2}$ for the resonator. Portions of the control pulses and solution are shown in \Cref{fig:controls-solution-cnot3}. We consider four initial conditions for the system, $\ket{000}$, $\ket{001}$, $\ket{010}$, and $\ket{011}$, which form the basis of the computational subspace.

\begin{figure}[htb!]
    \centering
    \includegraphics[width=\linewidth]{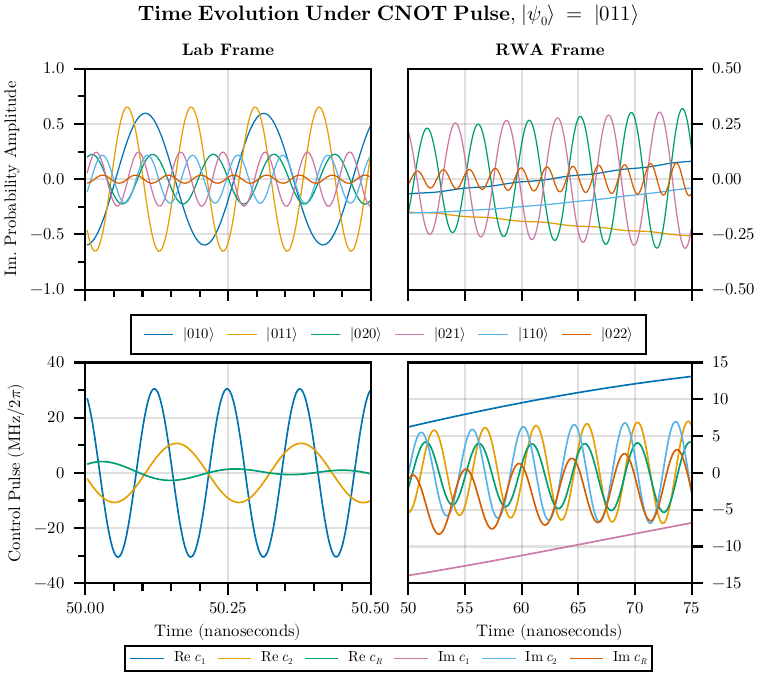}
    \caption{Lab- and RWA-frame time evolution for the CNOT gate simulation of \Cref{sec:multi-qudit-gate-problem}, starting from $\ket{\psib_0} = \ket{011}$. The RWA frame shows the window $50 \leq t \leq 75$ ns; the lab frame shows the much shorter window $50 \leq t \leq 50.5$ ns. Only the imaginary parts of the six largest probability amplitudes are shown.} 
    \label{fig:controls-solution-cnot3}
\end{figure}

As we did for the Rabi oscillator (\Cref{sec:rabi-oscillator-experiment}), we compare the Hermite, Filon, and Controlled Filon methods with $s = 0$, $1$, and $2$ by refining $\Delta t$ and computing the error against a high-order adaptive Runge--Kutta reference solution. For this more challenging example, we also record each method's runtime. As before, GMRES uses $10^{-15}$ tolerances, no preconditioning, and $\psib_n$ as the initial guess for $\psib_{n+1}$. The simulations ran on a single core of an Intel\textsuperscript{\textregistered} Xeon\textsuperscript{\textregistered} Gold 6148 CPU at 2.40 GHz. The results are shown in \Cref{fig:convergence-cnot3}; \Cref{tab:cnot3-speedup-time-basis} reports the speedup of the fastest method over the others at several target precisions.

\begin{figure}[htb!]
    \centering
    \includegraphics[width=\linewidth]{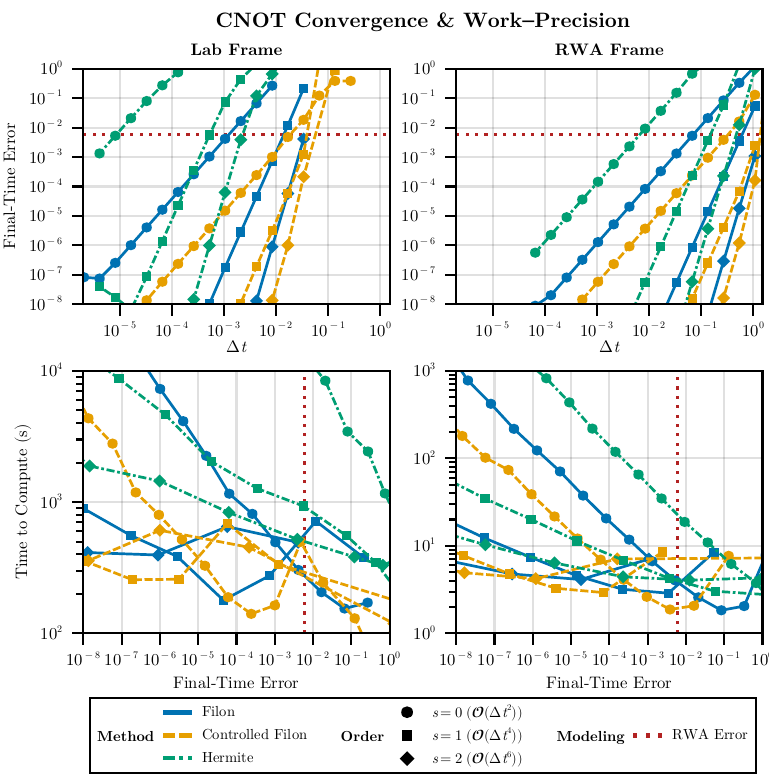}
    \caption{Convergence and work-precision of the Hermite, Filon, and Controlled Filon methods for the CNOT gate simulation of \Cref{sec:multi-qudit-gate-problem}, using the four computational-subspace basis states ($\ket{000}$, $\ket{001}$, $\ket{010}$, $\ket{011}$) as initial conditions. Final-time errors are measured using the Frobenius norm.}
    \label{fig:convergence-cnot3}
\end{figure}

\begin{table}
    \centering
\begin{tabular}{llrrrrrrr}
    \toprule
    \multicolumn{9}{c}{Relative Time to Compute Solution, Lab Frame} \\
    \midrule[\heavyrulewidth]
    & & \multicolumn{7}{c}{Target Final-Time Error} \\
    \cmidrule(lr){3-9}
    Method & $s$ & $10^{-1}$ & $10^{-2}$ & $10^{-3}$ & $10^{-4}$ & $10^{-5}$ & $10^{-6}$ & $10^{-7}$ \\
    \midrule
    Hermite & 0 & 2.3(1) & 7.7(1) & *2.1(2) & *5.6(2) & *1.2(3) & *3.7(3) & *1.1(4) \\
     & 1 & 3.7 & 5.9 & 8.0 & 9.3 & 1.0(1) & 1.9(1) & 3.0(1) \\
     & 2 & 2.8 & 3.4 & 4.3 & 4.7 & 4.2 & 5.6 & 6.1 \\
    \addlinespace
    Filon & 0 & 1.1 & 1.7 & 3.6 & 6.1 & 1.1(1) & 2.9(1) & 7.9(1) \\
     & 1 & 1.3 & 1.3 & 1.3 & 1.1 & 1.1 & 1.7 & 2.2 \\
     & 2 & 2.8 & 2.8 & 2.8 & 2.3 & 1.5 & 1.5 & 1.5 \\
    \addlinespace
    C-Filon & 0 & -- & -- & -- & -- & 1.5 & 3.1 & 7.2 \\
     & 1 & 1.6 & 1.8 & 1.8 & 1.5 & -- & -- & -- \\
     & 2 & 1.3 & 1.8 & 2.5 & 2.1 & 1.4 & 1.4 & 1.3 \\
    \bottomrule
\end{tabular}

    \caption{Relative time to simulate the CNOT gate of \Cref{sec:multi-qudit-gate-problem} to a desired lab-frame accuracy, starting from the computational-subspace basis states. The fastest method for each target error is marked ``--''. Each other entry reports the ratio of its time to the fastest method's (the ``speedup'' factor), to two significant digits. Asterisks mark extrapolated entries.}
    \label{tab:cnot3-speedup-time-basis}
\end{table}

\begin{figure}[htb!]
    \centering
    \includegraphics[width=\linewidth]{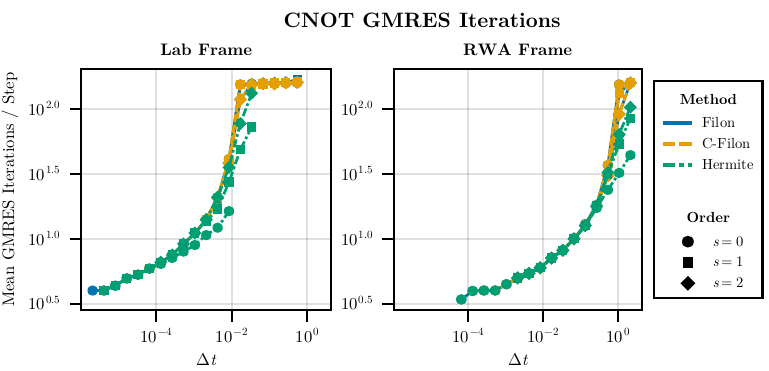}
    \caption{Average number of GMRES iterations per timestep for the Hermite, Filon, and Controlled Filon methods on the CNOT gate simulation of \Cref{sec:multi-qudit-gate-problem}.}
    \label{fig:cnot3-gmres}
\end{figure}

\subsection{Analysis} \label{sec:analysis}
The convergence plots in \Cref{fig:convergence-cnot3} show that the Filon and Controlled Filon methods are significantly more accurate than the Hermite method of the same order at the same stepsize, in both the RWA frame and (especially) the lab frame. For example, in the lab frame the $s=0$ Filon and Controlled Filon methods reach an accuracy of $10^{-2}$ using stepsizes that are approximately $100$ and $1000$ times, respectively, the stepsize required by $s=0$ Hermite.

However, these numbers do not necessarily mean that $s=0$ Filon and Controlled Filon will be $100$ and $1000$ times faster than $s=0$ Hermite. The cost per timestep differs across the three methods due to the additional $\oCal(N)$ operations that the Filon and Controlled Filon methods perform per application of $S^{s}_{\ExplicitAndImplicit}$ (\Cref{tab:matvec-counts}), as well as each method potentially requiring a different number of GMRES iterations to solve the linear system of equations at each timestep. The average number of GMRES iterations per timestep for each method is shown in \Cref{fig:cnot3-gmres}.

The runtime of all three methods could be improved by lowering the number of GMRES iterations per timestep (for example, through preconditioning). We expect such improvements to benefit Filon and Controlled Filon more than Hermite, which tends to require the fewest GMRES iterations here (\Cref{fig:cnot3-gmres}).

In this example, the stepsize advantage of the Filon and Controlled Filon methods translates to a significantly lower computational cost than Hermite. Indeed, \Cref{tab:cnot3-speedup-time-basis} and \Cref{fig:convergence-cnot3} show that in both the lab and RWA frames, the fastest method is always Controlled Filon. Furthermore, the $s=0$ Filon and Controlled Filon methods remain useful even at accuracies where the $s=0$ Hermite method is extremely expensive. For a target error of $10^{-4}$ in the lab frame, the fastest method is $s=0$ Controlled Filon, which is $4.7\times$ faster than the fastest Hermite method ($s=2$), and over $500\times$ faster than $s=0$ Hermite. The advantage grows as the target accuracy increases; for a target error of $10^{-7}$, $s=1$ Controlled Filon is fastest, being $6.1\times$ faster than $s=2$ Hermite, and $30\times$ faster than $s=1$ Hermite. When simulating in the RWA frame to the limit of the RWA model error (which is between $10^{-2}$ and $10^{-3}$), $s=0$ Controlled Filon is roughly $2\times$ faster than $s=2$ Hermite and $10\times$ faster than $s=0$ Hermite. Hence Controlled Filon is the most efficient method tested, both at the maximum true accuracy achievable in the RWA frame and, more generally, in the lab frame, where the dynamics can be simulated to even higher true accuracy.

\section{Conclusions and Future Work} \label{sec:conclusions}
In this work, we adapted Filon quadrature into two timestepping methods for linear ODE systems with highly oscillatory solutions: the Filon method (\Cref{sec:filon-for-ODEs}) and the Controlled Filon method (\Cref{sec:controlled-filon}). Both factor the solution into slowly varying envelopes and highly oscillatory carrier waves, and have the remarkable property that their local truncation error \emph{decreases} as the carrier wave frequencies \emph{increase} (\Cref{cor:filon-timestepping,cor:controlled-filon}), so a fixed accuracy is reached with far fewer timesteps than standard integrators require. The Controlled Filon method additionally accounts for the explicit carrier-wave structure of the control pulses.

We demonstrated the methods on two problems from quantum optimal control. For the Rabi oscillator (\Cref{sec:rabi-oscillator-experiment}), both methods attained the predicted $\oCal(\Delta t^{2(s+1)})$ convergence while reaching a given accuracy with substantially larger timesteps than the Hermite method of the same order. For a CNOT gate on two transmon qudits coupled to a resonator bus (\Cref{sec:multi-qudit-gate-problem}), the Controlled Filon method was the most efficient method in every case tested, outperforming the best Hermite method by up to $6\times$ and the Hermite method of the same order by up to $500\times$, depending on the frame and target accuracy.

Considering that high-order Hermite methods are already significantly faster than commonly used low-order methods ($3$--$100\times$ on the CNOT example, see \cite{HOHO_2025}), this is a remarkable improvement. Because their accuracy improves with frequency, the Filon methods make direct simulation in the lab frame affordable, avoiding the modeling error incurred by the rotating-wave approximation.

Several extensions remain. A straightforward goal is to reduce the per-timestep cost of the Filon methods through preconditioning, adaptive GMRES tolerances, and improved initial guesses. Because each timestep takes the form
$S_\Implicit \psib_{n+1} = S_\Explicit \psib_n$, the discrete-adjoint framework of
\cite{HOHO_2025} carries over directly, so the Filon methods can be used to compute exact, efficient gradients for gradient-based pulse optimization. Implementing and evaluating this adjoint is a natural next step. Another direction is to select the ansatz frequencies $\omegab$ adaptively, by estimating the solution frequency \emph{a posteriori} during the simulation, which would help when they cannot be inferred from the drift Hamiltonian or when they change over time. Finally, because the
methods apply to any linear ODE system with oscillatory solutions, they may be useful for problems outside the field of quantum optimal control.

\section*{Source Code and Data Availability}
Implementations of the Hermite, Filon, and Controlled Filon methods, as well as scripts for reproducing the figures and tables in this work, are available at \url{https://github.com/leespen1/FilonResearch}. The precollected simulation data used to produce the figures and tables is archived at \url{https://doi.org/10.5281/zenodo.21353408}.

\section*{Declarations}
AI tools were used to assist the development of the software and the revision of the manuscript. The authors assume responsibility for all content.

\appendix 

\section{Explicit Formulas for the Filon Method for Linear ODEs}
\label[appendix]{app:filon-method-S-formulas}
In this section, we construct explicit formulas for the matrices $S^s_\Implicit$ and $S^s_\Explicit$ which appear in the timestepping rule $S^s_\Implicit\psib_{n+1} = S^s_\Explicit\psib_n$ and result from \eqref{eq:filon-timestep-rule} in the Filon method that was introduced in \Cref{cor:filon-timestepping}. For the sake of compact notation, when convenient we will not write the time dependence of variables explicitly (e.g., we write $A$ instead of $A(t)$).

First, we construct matrices $D_m$ such that $D_m \psib = \psib^{(m)}$. This can be done using the general Leibniz rule on $\dot\psib(t) = A(t)\psib(t)$:
\begin{equation} \label{eq:state-derivatives-product-rule}
    \psib^{(m+1)} = \sum_{k=0}^m \binom{m}{k}A^{(m-k)}\psib^{(k)}.
\end{equation}
Computing the $\psib^{(m+1)}$ requires $m+1$ multiplications of vectors by $A(t)$ and its derivatives, so computing all derivatives up to $\psib^{(s)}$ requires $s(s+1)/2$ matrix-vector products. Expanding $\psib^{(k)}$ on the right-hand side of \eqref{eq:state-derivatives-product-rule} gives us $D_{m+1}$. The first four instances of $D_m$ are
\begin{equation*}
    D_0 = I,\quad D_1 = A,\quad D_2 = \dot{A} + A^2,\quad D_3 = \ddot{A} + 2\dot{A}A + A\dot{A} + A^3.
\end{equation*}
We also define matrices $F_m$ such that $RF_m\psib = \boldf^{(m)}$. Using the general Leibniz rule on $\boldf(t) = R(t)\psib(t)$,
\begin{equation} \label{eq:F-matrix-definition}
    \boldf^{(m)}
    = \sum_{k=0}^m \binom{m}{k}R^{(m-k)}\psib^{(k)}
    = R F_m \psib,\quad F_m \coloneq
    \sum_{k=0}^m \binom{m}{k}(-\ii \Omega)^{m-k} D_k,
\end{equation}
where $\Omega \coloneq \diag(\omegab)$ ($R$ and $\Omega$ commute because they are both diagonal). The first three instances of $F_m$ are
\begin{alignat*}{2}
    F_0 &= D_0
        &{}={}& I \\
    F_1 &= (-\ii\Omega)D_0 + D_1 
        &{}={}& -\ii\Omega + A \\
    F_2 &= (-\ii\Omega)^2D_0 + 2(-\ii\Omega)D_1 + D_2
        &{}={}& -\Omega^2 - 2\ii\Omega A + (\dot{A} + A^2).
\end{alignat*}
We did not include the $R$ factor in the definition of $F_m$ so that we can instead ``absorb'' it into the weight matrices $B^s_{\ExplicitOrImplicit,j}(\omegab)$, defining the diagonal \emph{weight-phase} matrices
\begin{equation*}
W^s_{\ExplicitAndImplicit, j}
\coloneq\; B^s_{\ExplicitAndImplicit,j}(\omegab)R(\tExpImp)
\;=\; \left(\frac{\Delta t}{2}\right)^{j+1} \diag \left( \left\{
    e^{\pm\ii\hat\omega_k}
    \, b^{[-1,1],s}_{\ExplicitAndImplicit,j} \!\left(\hat\omega_k\right)
\right\}_{k=1}^N \right),
\end{equation*}
where $\hat\omega_k \coloneq \omega_k \Delta t/2$. Then we have
\begin{equation*}
    \dv[j]{}{t}\Big( A B^s_{\ExplicitOrImplicit, j}(\omegab) \boldf \Big)
    = \sum_{k=0}^j \binom{j}{k} A^{(j-k)}B^s_{\ExplicitOrImplicit,j}(\omegab) \boldf^{(k)}
    = \sum_{k=0}^j \binom{j}{k}A^{(j-k)}W^s_{\ExplicitOrImplicit,j}F_k\psib.
\end{equation*}
This finally allows us to write
\begin{equation*}
    S^s_\ExplicitAndImplicit = I \pm \left[\sum_{j=0}^s\sum_{k=0}^j \binom{j}{k} A^{(j-k)}W^s_{\ExplicitAndImplicit,j}F_k\right]_{\tExpImp} 
    \!\! = I \pm \left[\sum_{i=0}^s A^{(i)} \sum_{j=i}^s \binom{j}{i} W^s_{\ExplicitAndImplicit,j} F_{j-i}\right]_{\tExpImp}\!.
\end{equation*}
When computing $S^s_{\ExplicitAndImplicit}\psib$, the second ordering above allows us to apply each $A^{(i)}$ only once, to the output of the inner summation. Moreover, each $F_{m}\psib$ can be computed once and reused across the outer summation. For $s=0,1,2$, the matrices are
\begin{align*}
    S^0_\ExplicitAndImplicit &= I \pm \big[ A W^0_{\ExplicitAndImplicit,0} \big]_{\tExpImp}, \\
    S^1_\ExplicitAndImplicit &= I \pm \big[ A W^1_{\ExplicitAndImplicit,0}
        + \dot{A} W^1_{\ExplicitAndImplicit,1}
        + A W^1_{\ExplicitAndImplicit,1}(A - \ii\Omega) \big]_{\tExpImp}, \\
    S^2_\ExplicitAndImplicit &= I \pm \big[ A W^2_{\ExplicitAndImplicit,0}
        + \dot{A} W^2_{\ExplicitAndImplicit,1}
        + \ddot A W^2_{\ExplicitAndImplicit,2}
        + A W^2_{\ExplicitAndImplicit,1}(A-\ii\Omega) \\
        &\quad\quad + 2\dot{A} W^2_{\ExplicitAndImplicit,2}(-\ii\Omega+A)
        + A W^2_{\ExplicitAndImplicit,2}\left( -\Omega^2 - 2\ii \Omega A + \dot{A} + A^2\right) \big]_{\tExpImp}.
\end{align*}
When we choose the ansatz $\omegab = [0,\dots,0]^T$, the Filon method reduces to the Hermite method. In this case, $\Omega = 0$ and each $W^s_{\ExplicitAndImplicit,j}$ can be replaced with a scalar weight, reducing the number of diagonal matrix-vector multiplications performed.

Naively, when $A(t) = \sum_{k=1}^{\nHam} c_k(t)A_k$ as in \eqref{eq:scalar-controlled-hamiltonian-A}, the number of matrix-vector products is multiplied by a factor of $\nHam$. However, we can make further optimizations that \emph{reduce} the scaling in $s$, from $\frac{(s+1)(s+2)}{2}\nHam$ to $(s+1)\nHam$. Factoring the constant $A_k$ matrices out of $A^{(j-k)}$, we get
\begin{equation*}
    S^s_\ExplicitAndImplicit
    = I \pm \left[ \sum_{k=1}^{\nHam} A_k \sum_{j=0}^s \sum_{m=0}^j
        \binom{j}{m} \ctrlFunc_k^{(j-m)} W^s_{\ExplicitAndImplicit,j} F_m \right]_{\tExpImp},
\end{equation*}
which reduces the number of matrix-vector products in the outermost summation from $(s+1)\nHam$ to $\nHam$. Applying the same factorization to \eqref{eq:state-derivatives-product-rule} reduces the number of matrix-vector products required to compute $\psib^{(0)},\ldots,\psib^{(s)}$ from $\frac{s(s+1)}{2}\nHam$ to $s\nHam$.

Swapping the order of the $j$ and $m$ summations and gathering the weight-phase matrices $W^s_{\ExplicitAndImplicit,j}$ into the time-dependent diagonal matrix $G^s_{\ExplicitAndImplicit,k,m}$ (which can be reused across GMRES iterations) yields the efficient Filon form
\begin{equation} \label{eq:efficient-filon-S-formula}
    S^s_\ExplicitAndImplicit = I \pm \sum_{k=1}^{\nHam} A_k \sum_{m=0}^s
        \left[ G^s_{\ExplicitAndImplicit,k,m} F_m \right]_{\tExpImp}\!\!, \!\quad
    G^s_{\ExplicitAndImplicit,k,m} \coloneq \sum_{j=m}^s \binom{j}{m} \ctrlFunc_k^{(j-m)} W^s_{\ExplicitAndImplicit,j}.
\end{equation}

\section{Explicit Formulas for the Controlled Filon Method}
\label[appendix]{app:controlled-filon-method-S-formulas}
In \Cref{app:filon-method-S-formulas}, we gave explicit formulas for the matrices $S_\Implicit^s$ and $S_\Explicit^s$ for the timestepping rule $S_\Implicit^s\psib_{n+1} = S_\Explicit^s \psib_n$ for the Filon method introduced in \Cref{cor:filon-timestepping}. In this section, we give analogous matrices for the timestepping rule for controlled quantum systems, introduced in \Cref{cor:controlled-filon}. 

First, we expand the derivatives in the summand of \eqref{eq:controlled-filon-timestep-system-of-equations}:
\begin{equation*}
    B^s_{\ExplicitOrImplicit,j}(\tilde{\omegab}_{k,l}) ( \tilde{\ctrlFunc}_{k,l}\boldf )^{(j)}
\!=\! B^s_{\ExplicitOrImplicit,j}(\tilde{\omegab}_{k,l}) \!\!\sum_{m=0}^j \! \binom{j}{m} \tilde{\ctrlFunc}_{k,l}^{(j-m)}\boldf^{(m)}
\!=\! \sum_{m=0}^j \! \binom{j}{m}\tilde{\ctrlFunc}_{k,l}^{(j-m)}B^s_{\ExplicitOrImplicit,j}(\tilde{\omegab}_{k,l}) \boldf^{(m)}.
\end{equation*}
Next, we wish to rewrite $B^s_{\ExplicitOrImplicit,j}(\tilde{\omegab}_{k,l}) \boldf^{(m)}$ in terms of $\psib$.  Recall from \eqref{eq:diagonal-weight-matrices} that
\begin{equation*}
    B^s_{\ExplicitOrImplicit,j}(\tilde{\omegab}_{k,l}) \coloneq
        \left(\frac{\Delta t}{2}\right)^{j+1}
        \diag \left( \left\{
            e^{\ii (\omega_i + \nu_{k,l}) \overline{t}_n}b^{[-1,1],s}_{\ExplicitOrImplicit,j}\left((\omega_i+\nu_{k,l}) \frac{\Delta t}{2}\right)
        \right\}^{N}_{i=1}\right),
\end{equation*}
where $\overline{t}_n \coloneq (t_n+t_{n+1})/2$ is the midpoint of the timestep interval. Then we have
\begin{equation*}
    B^s_{\Implicit,j}(\tilde{\omegab}_{k,l})R(t_{n+1}) = e^{\ii\nu_{k,l}\overline{t}_n}
    \left(\frac{\Delta t}{2}\right)^{j+1} 
    \!\!\!\! \diag \! \left( \! \left\{
        e^{-\ii\hat{\omega}_i }b^{[-1,1],s}_{\Implicit,j}\left((\omega_i + \nu_{k,l})\frac{\Delta t}{2}\right)
    \right\}^N_{i=1} \right).
\end{equation*}

Similar to \Cref{app:filon-method-S-formulas}, we define diagonal weight-phase matrices $W^s_{\ExplicitAndImplicit,j,k,l}$, where this time we use additional indices to account for the modification of the frequencies $\omegab$ by the carrier frequency $\nu_{k,l}$ of each carrier:
\begin{equation*}
    W^s_{\ExplicitAndImplicit,j,k,l} \coloneq
        \left(\frac{\Delta t}{2}\right)^{j+1}
        \diag \left( \left\{
            e^{\pm\ii \hat\omega_i } b^{[-1,1],s}_{\ExplicitAndImplicit,j}\left((\omega_i+\nu_{k,l}) \frac{\Delta t}{2}\right)
        \right\}^{N}_{i=1}\right).
\end{equation*}
We note that the weight-phase matrices are time-independent. Using the above definitions, and the matrices $F_m$ of \eqref{eq:F-matrix-definition}, we have
\begin{equation*}
   e^{\ii \nu_{k,l}\overline{t}_n} W^s_{\ExplicitOrImplicit,j,k,l}F_m \psib = B^s_{\ExplicitOrImplicit,j}(\tilde{\omegab}_{k,l})\boldf^{(m)}.
\end{equation*}
This allows us to write $S^s_\Implicit$ and $S^s_\Explicit$ for the Controlled Filon method as
\begin{equation*}
    S^s_\ExplicitAndImplicit = I \pm \left[
    \sum_{k=1}^{\nHam} A_k \sum_{l=1}^{N_{\mathrm{F},k}} \phi_{\ExplicitAndImplicit,k,l} \sum_{j=0}^s\sum_{m=0}^j
        \binom{j}{m}\tilde{\ctrlFunc}_{k,l}^{(j-m)}W^s_{\ExplicitAndImplicit,j,k,l}F_m
    \right]_{\tExpImp},
\end{equation*}
where the time-dependent phase factor $\phi_{\ExplicitAndImplicit,k,l}$ is defined by
\begin{equation*}
    \phi_{\ExplicitAndImplicit, k,l}(t) \coloneq  e^{\ii \nu_{k,l}(t \pm \frac{\Delta t}{2})}.
\end{equation*}
We have factored the matrix $A_k$ out of the sum over carriers of the same control. Exchanging the order of the $j$ and $m$ summations and gathering every carrier of a given control into a single time-dependent diagonal matrix,
\begin{equation} \label{eq:controlled-filon-G-definition}
    G^s_{\ExplicitAndImplicit,k,m} \coloneq \sum_{l=1}^{N_{\mathrm{F},k}} \phi_{\ExplicitAndImplicit,k,l} \sum_{j=m}^s \binom{j}{m} \tilde{\ctrlFunc}_{k,l}^{(j-m)} W^s_{\ExplicitAndImplicit,j,k,l},
\end{equation}
yields the more efficient Controlled Filon form
\begin{equation} \label{eq:explicit-controlled-filon-efficient}
    S^s_\ExplicitAndImplicit = I \pm \sum_{k=1}^{\nHam} A_k \sum_{m=0}^s \left[ G^s_{\ExplicitAndImplicit,k,m} F_m \right]_{\tExpImp}.
\end{equation}
As with the Filon method, each $F_m\psib$ can be computed once and reused. Outside of computing $F_m\psib$, each 
matrix $A_k$ is applied only once, in the outer summation. Therefore, the
number of non-diagonal matrix multiplications performed does not grow with the
number of carriers $N_{\mathrm{F},k}$.

\bibliographystyle{siamunsrt}
\bibliography{bibliography}
\end{document}